\documentclass[11pt]{article}
\usepackage[utf8]{inputenc}
\usepackage{amsfonts}
\usepackage{braket}
\usepackage[para]{footmisc}
\usepackage{xcolor}
\usepackage{qcircuit}
\usepackage{amsmath}
\usepackage{enumerate}
\usepackage[pagebackref]{hyperref}
\hypersetup{
   colorlinks,
   linkcolor={red!100!black},
   citecolor={green!100!black},
}
\usepackage{color, colortbl}
\definecolor{Gray}{gray}{0.85}
\usepackage{multirow}
\usepackage{makecell}
\usepackage{bbm}
\usepackage{amsthm,mathrsfs}
\usepackage{tikz}
\usepackage{mleftright}
\usepackage{amssymb}
\usepackage{algorithm}
\usepackage{algpseudocode}
\algrenewcommand\algorithmicrequire{\textbf{Input:}}
\algrenewcommand\algorithmicensure{\textbf{Output:}}

\algnewcommand{\IIf}[1]{\State\algorithmicif\ #1\ \algorithmicthen}
\algnewcommand{\EndIIf}{\unskip\ \algorithmicend\ \algorithmicif}
\usepackage{natbib}
\usepackage{mathabx}
\usepackage{soul}
\usepackage{geometry}
\usepackage{comment}
\usepackage{authblk}

\newtheorem{theorem}{Theorem}
\newtheorem{lemma}{Lemma}

\newtheorem{fact}{Fact}
\newtheorem{result}{Result}

\newtheorem{corollary}[theorem]{Corollary}

\usepackage[capitalize,nameinlink,noabbrev]{cleveref}

\Crefname{lemma}{Lemma}{Lemmas}
\Crefname{fact}{Fact}{Facts}

\usepackage{mathtools}

\DeclareMathOperator*{\argmax}{arg\,max}

\def\be{\begin{eqnarray}}
\def\ee{\end{eqnarray}}

\newcommand{\poly}{\operatorname{poly}}

\begin{document}

\title{Improved Quantum Algorithms for Reinforcement Learning Under a Generative Model}
\author[1]{Joao F. Doriguello\thanks{\tiny doriguello@renyi.hu}}

\affil[1]{HUN-REN Alfr\'ed R\'enyi Institute of Mathematics, Budapest, Hungary}
\date{\today}
\maketitle

\begin{abstract} 
    Reinforcement learning is a subfield of machine learning that studies how an agent interacts with an environment in order to extract as large a reward as possible. A standard approach to study such interaction is through Markov Decision Processes (MDPs) and the task of choosing an optimal policy---a function that tells the agent which action to take. In this work, we study two types of time-dependent MDP---finite-horizon and infinite-horizon discounted---and propose new quantum algorithms for computing  approximate optimal policies. Our quantum algorithms are based on a new combination of standard value iteration and quantum subroutines like quantum mean estimation and quantum maximum finding, overall enhanced with techniques from sample-optimal classical algorithms. Our resulting query complexities improve upon previous works, thus approaching already established quantum lower bounds.
\end{abstract}

\section{Introduction}
\label{sec:intro}

The paradigm of reinforcement learning (RL)~\citep{sutton1998reinforcement} has received much attention recently due to its success in modeling an agent interaction with a dynamical environment~\citep{aastrom1965optimal,hu2007markov,sato2010markov,bauerle2011markov,feinberg2012handbook,bennett2013artificial,chen2014distributed,steimle2017markov,natarajan2022planning} and due to its applicability to a wide range of problems~\citep{sutton1998reinforcement,Szepesvari2010algorithms,bertsekas2012dynamic,bertsekas2022abstract}. Such interaction is usually studied using Markov Decision Processes (MDPs)~\citep{puterman2014markov}, a simple yet powerful mathematical abstraction wherein an agent chooses actions given the environment's state in order to maximise some kind of reward. Despite their versatility, MDPs suffer a ``curse of dimensionality'' when the number of possible actions or states are sufficiently large, rendering a solution computationally impossible~\citep{powell2007approximate}.

Parallel to RL developments, quantum computation~\citep{nielsen2010quantum} has emerged as a new field promised to deliver substantial speedups to traditional problems like factoring~\citep{shor1994algorithms} and unstructured search~\citep{grover1996fast} by exploiting quantum mechanics. Quantum machine learning~\citep{Schuld2015introduction,Biamonte2017quantum,allcock2019quantum,Alchieri2021introduction}, a subfield of quantum computation, has been consolidated as an area of study that can potentially provide quantum speedups to several traditional machine learning problems~\citep{lloyd2013quantum,rebentrost2014quantum,wiebe2015quantum,kerenidis2017quantum,kerenidis2019qmeans,liu2021rigorous,doriguello2023you}, including RL~\citep{dong2008quantum,crawford2018reinforcement,Lumbreras2022multiarmedquantum,dai2023quantum,wan2023quantum,su2025quantum,nagy2025hybrid,ambainis2025bit}. At the same time, new advances in quantum hardware~\citep{Arute2019quantum,Saggio2021experimental,Acharya2025quantum} point to a future when quantum computers can become a reality.

In this work, we propose novel quantum algorithms to tackle a central problem related to MDPs, obtaining approximately optimal policies, thus alleviating the curse of dimensionality.

\subsection{Problem Setup}

In its most general form, a time-dependent MDP at time step $t\in\mathbb{N}$ is described by a finite set $\mathcal{S}$ of states that the environment can assume and a finite set $\mathcal{A}$ of actions that the agent can take, a reward function $r_t:\mathcal{S}\times\mathcal{A}\to[0,1]$ that describes the agent's reward when choosing action $a\in\mathcal{A}$ under the environment's state $s\in\mathcal{S}$, and stochastic kernels $p_t(\cdot|s,a)$ that denote the probability of transitioning to a new state after the reward is obtained. Beyond this general formalisation, there are different subtypes of MDPs depending on the reward criteria the agent must maximise. Central to this work are \emph{finite-horizon} MDPs, where the agent interacts with the environment for a pre-determined number of steps $H$ called \emph{horizon}, and \emph{infinite-horizon discounted} MDPs, where the agent-environment interaction lasts an infinite amount of time but at every time step the received reward is decreased by a \emph{discount factor} $\gamma\in[0,1)$. The parameter $\Gamma := (1-\gamma)^{-1}$ is usually called \emph{effective horizon}. Finite- and infinite-horizon discounted MDPs are thus described by tuples $\langle \mathcal{S},\mathcal{A},\{p_h\}_h,\{r_h\}_h,H\rangle$ and $\langle \mathcal{S},\mathcal{A},\{p_t\}_t,\{r_t\}_t,\gamma\rangle$, respectively.

The main agent's objective when interacting with the environment is to maximise some reward criteria by choosing an appropriate \emph{policy} $\pi = (\pi_t)_t$, which is a sequence of probability distributions over $\mathcal{A}$ given $\mathcal{S}$ called \emph{decision rules}, i.e., $\pi_t : \mathcal{S}\to\Delta(\mathcal{A})$. The reward criteria for finite-horizon MDPs is the expected sum of rewards under policy $\pi = (\pi_h)_{h\in[H]}$ and initial state $s\in\mathcal{S}$, called \emph{value function}:
\begin{align*}
    V_0^{\pi,H}(s) = \mathbb{E}\left[\sum_{h=0}^{H-1} r_h(s_h,a_h) ~\Bigg|~ s_0 = s, a_i \sim \pi_i(s_i), s_{i+1} \sim p_i(\cdot|s_i, a_i)\right].
\end{align*}
A policy $\pi^\varepsilon$ is $\varepsilon$-\emph{optimal} if $V_0^{\pi^\varepsilon,H}(s) \geq V_0^{\pi,H}(s) - \varepsilon$ for all $\pi$ and $s\in\mathcal{S}$, and the \emph{optimal value function} is $V_0^{\ast,H}(s) = \sup_\pi V_0^{\pi,H}(s)$. A function $V:\mathcal{S}\to[0,H]$ is \emph{$\varepsilon$-optimal} if $\|V - V_0^{\ast,H}\|_\infty \leq \varepsilon$.

On the other hand, the reward criteria for infinite-horizon MDPs is the discounted sum of rewards under policy $\pi = (\pi_t)_{t\in\mathbb{N}}$ and initial state $s\in\mathcal{S}$, also called \emph{value function}:
\begin{align*}
    V_0^{\pi,\gamma}(s) = \mathbb{E}\left[\sum_{t=0}^\infty \gamma^t r_t(s_t,a_t) ~\Bigg|~ s_0 = s, a_i \sim \pi_i(s_i), s_{i+1} \sim p_i(\cdot|s_i, a_i)\right].
\end{align*}
A policy $\pi^\varepsilon$ is \emph{$\varepsilon$-optimal} if $V_0^{\pi^\varepsilon,\gamma}(s) \geq V_0^{\pi,\gamma}(s) - \varepsilon$ for all $\pi$ and $s\in\mathcal{S}$, and the \emph{optimal} value function is $V_0^{\ast,\gamma}(s) = \sup_\pi V_0^{\pi,\gamma}(s)$. A function $V:\mathcal{S}\to[0,\Gamma]$ is \emph{$\varepsilon$-optimal} if $\|V - V_0^{\ast,\gamma}\|_\infty \leq \varepsilon$.

One of the fundamental theoretical problems related to MDPs is that of learning an $\varepsilon$-optimal policy, i.e., obtaining $\pi_\varepsilon\in\Pi^{\rm D}$ such that $V_0^{\pi_\varepsilon,H}(s) \geq V_0^{\ast,H}(s) - \varepsilon$ (finite-horizon) or  $V^{\pi_\varepsilon,\gamma}_0(s) \geq V_0^{\ast,\gamma}(s) - \varepsilon$ (infinite-horizon) for all $s\in\mathcal{S}$. Among the several studied input models, arguably the most important---and the one considered in this work---is the so-called \emph{generative model}~\citep{kearns1998finite,Kearns2002sparse,kakade2003sample} where one has full knowledge of state and action spaces and of the reward functions, but the transition probabilities $p_t(\cdot|s,a)$ can only be accessed through an oracle.

\paragraph{Classical setting.} In the classical setting using standard classical computers, generative access to the stochastic kernel $p_t:\mathcal{S}\times\mathcal{A}\to\Delta(\mathcal{S})$ is given via an oracle or simulator $\mathcal{C}_{p_t}$ that, on input $(s,a)\in\mathcal{S}\times\mathcal{A}$, returns $s'\in\mathcal{S}$ with probability $p_t(s'|s,a)$. One is thus interested in computing an $\varepsilon$-optimal policy with the least amount of queries to oracles $\{\mathcal{C}_{p_t}\}_t$ (there is only one oracle for time-independent MDPs). 

\paragraph{Quantum setting.} The quantum equivalent of generative access to $p_t$ is done through an oracle $\mathcal{Q}_{p_t}$ (and its inverse $\mathcal{Q}_{p_t}^\dagger$) called quantum accessible-environment~\citep{wang2021quantum,wiedemann2022quantum, jerbi2022quantum, zhong2023provably} which is the unitary operator
\begin{align*}
    \mathcal{Q}_{p_t}: |s\rangle|a\rangle|\bar 0\rangle\rightarrow \sum_{s'\in\mathcal S} \sqrt{p_t(s'\vert s, a)} |s\rangle|a\rangle|s'\rangle|\mathrm{garbage}(s,a,s')\rangle \qquad \forall (s,a)\in\mathcal{S}\times\mathcal{A},
\end{align*}
where $|\mathrm{garbage}(s,a,s')\rangle$ are normalised states. The quantum oracle $\mathcal{Q}_{p_t}$ puts in superposition all the possible elements $s'\in\mathcal{S}$ with the corresponding weight $\sqrt{p_t(s'|s,a)}$ and is a proper quantum generalisation of $\mathcal{C}_{p_t}$ by noticing that a measurement on the third register $|s'\rangle$ returns an element $s'\in\mathcal{S}$ with probability $p_t(s'|s,a)$, just like in the classical setting. Moreover, if the classical oracle $\mathcal{C}_{p_t}$ is itself a computer program and we have access to its source code, then it is straightforward to transform $\mathcal{C}_{p_t}$ into its quantum version $\mathcal{Q}_{p_t}$ above: one first produces a Boolean circuit (of roughly the same size as the program's time complexity) that produces the samples $s'\sim p_t(\cdot|s,a)$ and then replaces all classical gates with their quantum counterparts as described, for example, in~\cite[Section~1.4.1]{nielsen2010quantum}.

\subsection{Previous Works}

Classically computing $\varepsilon$-optimal policies under the generative model has been mostly solved for finite- and infinite-horizon discounted MDPs. A long list of works slowly improved the query complexity~\citep{Kearns1999finite,Kearns2002sparse,GheshlaghiAzar2013,wang2017randomized,Sidford2018variance} until Sidford et al.~\cite{sidford2018near} obtained sample-optimal algorithms for \emph{time-independent} finite- and infinite-horizon discounted MDPs with query complexity $\widetilde{O}\big(\frac{H^3SA}{\varepsilon^2} \big)$ and $\widetilde{O}\big(\frac{\Gamma^3SA}{\varepsilon^2}\big)$, respectively, thus matching the lower bounds from~\cite{GheshlaghiAzar2013,sidford2018near}. Later Li et al.~\cite{li2020breaking} obtained the query upper bound $\widetilde{O}\big(\frac{H^4SA}{\varepsilon^2} \big)$ for time-dependent finite-horizon MDPs.

The algorithms of Sidford et al.~\cite{sidford2018near} are a modern and improved version of the standard value iteration algorithm~\citep{puterman2014markov} (explained in \Cref{sec:results}) and work (e.g., for infinite-horizon discounted MDPs) by starting with a function $V_0$ such that $0 \leq V^{\ast,\gamma}_0(s) - V_0(s) \leq 2\epsilon$ and then producing function $V_T$ such that $0 \leq V_0^{\ast,\gamma}(s) - V_T(s) \leq \epsilon$, thus halving the initial error. The initial choice $V_0(s) = 0$ yields a $\Gamma$-approximation and by repeating their halving procedure $O\big(\Gamma\log\frac{\Gamma}{\varepsilon}\big)$ times, an $\varepsilon$-optimal policy can be obtained.

In the quantum setting where one has access to the oracle $\mathcal{Q}_p$ and its inverse, Wang et al.~\cite{wang2021quantum} were the first to tackle the problem of approximating optimal policies in the time-independent infinite-horizon discounted setting and obtained the query upper bound $\widetilde{O}\big({\min}\big\{\frac{\Gamma^{1.5}SA}{\varepsilon},\frac{\Gamma^3S\sqrt{A}}{\varepsilon}\big\}\big)$ and the query lower bound $\widetilde{\Omega}\big(\frac{\Gamma^{1.5}S\sqrt{A}}{\varepsilon}\big)$. The time-independent finite-horizon setting, on the other hand, has been more recently addressed by Ambainis et al.~\cite{ambainis2025bit}, and the time-dependent setting by Luo et al.~\cite{luo2025quantum}, both works proposing query upper bounds of $\widetilde{O}\big({\min}\big\{\frac{H^{2.5}SA}{\varepsilon},\frac{H^3S\sqrt{A}}{\varepsilon}\big\}\big)$. Luo et al.~\cite{luo2025quantum} further proved a query lower bound of $\widetilde{\Omega}\big(\frac{H^{1.5}S\sqrt{A}}{\varepsilon}\big)$. We are not aware of any works, classical or quantum, that considered time-dependent infinite-horizon discounted MDPs.

The quantum algorithms of~\cite{wang2021quantum,luo2025quantum,ambainis2025bit} are the combination of two different algorithms. The first one is a quantised version of the standard value iteration~\citep{puterman2014markov}. The quantum advantage comes from estimating quantities like $\sum_{s'\in\mathcal{S}} p_t(s'|s,a)V_{t}(s')$ with improved precision using a quantum mean estimation subroutine~\citep{montanaro2015quant,hamoudi2019quantum,hamoudi2021quantum2,hamoudi2021quantum,Kothari2023mean} and nesting it into the quantum maximum finding subroutine of \cite{durr1996quantum} in order to find a maximum over $a\in\mathcal{A}$ quadratically faster.

The second algorithm of~\cite{wang2021quantum,luo2025quantum,ambainis2025bit} is a quantised version of the modern value iteration algorithm of Sidford et al.~\cite{sidford2018near} mentioned above. The quantum advantage again comes from estimating mean quantities with improved precision using quantum mean estimation subroutines, which leads to a better $H$ or $\Gamma$ dependence. However, due to the algorithm's complexity of~\cite{sidford2018near}, it was not possible to employ quantum maximum finding in order to further improve the complexity on $A$. This is partially due to the fact that in~\cite{sidford2018near}, several average quantities are estimated at the beginning of a halving step. Pre-processing all these quantities in the quantum setting already costs linearly in $A$, hindering the quadratic advantage on the action space size. 
The combination of both quantised algorithms just highlighted above is the reason why the query upper bounds of~\cite{wang2021quantum,luo2025quantum,ambainis2025bit} are a minimum of two different quantities, one with better $H,\Gamma$ complexity and the other with better $A$ complexity.

\subsection{Our Results}
\label{sec:results}

As just explained and summarised in \Cref{table:results}, there is still a considerable gap between lower and upper quantum query complexities for the problem of approximating optimal policies. Interesting enough, the quantum algorithms of \cite{wang2021quantum,luo2025quantum,ambainis2025bit} can only provide an improvement either in $H,\Gamma$ or in $A$ compared to their classical counterparts, but never jointly in both parameters. As our main contribution, we propose new quantum algorithms that improve upon the previous query upper bounds for finite- and infinite-horizon discounted MDPs.

\begin{result}\label{res:res1}
    Let $\langle\mathcal{S},\mathcal{A},\{p_h\}_{h},\{r_h\}_{h},H\rangle$ be a time-dependent finite-horizon MDP. There is a quantum algorithm that outputs an $\varepsilon$-optimal policy w.h.p.\ and query complexity $\widetilde{O}\big(\frac{H^{2.5}S\sqrt{A}}{\varepsilon}\big)$.
\end{result}

\begin{result}\label{res:res2}
    Let $\langle\mathcal{S},\mathcal{A},\{p_t\}_{t},\{r_t\}_{t},\gamma\rangle$ be a time-dependent infinite-horizon discounted MDP. There is a quantum algorithm that outputs an $\varepsilon$-optimal policy w.h.p.\ and query complexity $\widetilde{O}\big(\frac{\Gamma^{2.5}S\sqrt{A}}{\varepsilon}\big)$.
\end{result}

Our new query upper bounds provide a joint quantum advantage in both $H,\Gamma$ and $A$ parameters, thus breaking the barrier of previous algorithms \citep{wang2021quantum,luo2025quantum,ambainis2025bit} and approaching the corresponding query lower bounds. Moreover, \Cref{res:res2} applies to the time-dependent setting for infinite-horizon discounted MDPs, which has not been addressed before. More surprisingly, we are not aware of any classical works that considered the optimal-policy problem for time-dependent infinite-horizon discounted MDPs under the generative model. 

Our quantum algorithms employ the quantum mean estimation subroutine \emph{with variance} from Kothari and O'Donnell~\cite{Kothari2023mean} (\Cref{fact:quantum_mean_estimation_variance}) within a generalised version of quantum maximum finding from van Apeldoorn et al.~\cite{vanApeldoorn2020quantumsdpsolvers} (\Cref{fact:quantum_minimum_finding}), plus techniques like monotonicity and total-variance from~\cite{sidford2018near}. In contrast to previous quantum works that either (i) use quantum maximum finding but does not include any variance in the error propagation or (ii) take variance into account but not quantum maximum finding, here we perform both, which ultimately yields a combined dependence on $\sqrt{A}$ and a better dependence than $\Gamma^3$ or $H^3$. Our finite-horizon complexity thus subsumes the ones of~\cite{luo2025quantum,ambainis2025bit}, while our infinite-horizon complexity improves the part $\widetilde{O}\big(\frac{\Gamma^3S\sqrt{A}}{\varepsilon}\big)$ from Wang et al.~\cite{wang2021quantum}, leading to a new overall complexity of $\widetilde{O}\big({\min}\big\{\frac{\Gamma^{1.5}SA}{\varepsilon},\frac{\Gamma^{2.5}S\sqrt{A}}{\varepsilon}\big\}\big)$. Our results are thus a measurable improvement
on past works in tackling the curse of dimensionality of MDPs.

We now give a high-level summary of our quantum algorithms. The formal proof of \Cref{res:res1} is given in \Cref{sec:finite_horizon}, while the proof of \Cref{res:res2} is given in \Cref{app:infinite_horizon}.


\begin{table*}[t]
\centering
\caption{Classical and quantum query bounds for computing an $\varepsilon$-optimal policy and value function for finite- and infinite-horizon discounted MDPs in the generative model. Here $S$ and $A$ are the sizes of the state and action spaces, respectively, $H$ is the horizon, and $\Gamma = (1-\gamma)^{-1}$ the effective horizon. All bounds are up to $\poly\log$ factors and assume constant failure probability.}
\def\arraystretch{2.5}
\resizebox{\linewidth}{!}{
\begin{tabular}{|c|cc|ccc|}
\hline

\multirow{2}{*}{Setting}  & \multicolumn{2}{c|}{Classical query complexity} & \multicolumn{3}{c|}{Quantum query complexity} \\ 
\cline{2-6}

& \multicolumn{1}{c|}{Upper bound} & Lower bound & \multicolumn{2}{c|}{Upper bound} & Lower bound \\ \hline\hline

\makecell{Time-independent \\ finite-horizon} & \multicolumn{1}{c|}{\makecell[c]{\Large $\frac{H^3SA}{\varepsilon^2}$ \\ \cite{sidford2018near}}} & \makecell[c]{\Large$\frac{H^3SA}{\varepsilon^2}$ \\ \cite{sidford2018near}} & \multicolumn{1}{c|}{\makecell[c]{\Large$\min\!\left\{\!\frac{H^{2.5}SA}{\varepsilon},\frac{H^3S\sqrt{A}}{\varepsilon}\right\}$ \\ \cite{ambainis2025bit,luo2025quantum}}} & \multicolumn{1}{c|}{\cellcolor{blue!20}\mbox{\makecell[c]{\Large$\frac{H^{2.5}S\sqrt{A}}{\varepsilon}$ \\ \Cref{res:res1}}}} & \makecell[c]{\Large$\frac{H^{1.5}S\sqrt{A}}{\varepsilon}$ \\ \cite{luo2025quantum} } \\ \hline

\makecell{Time-dependent \\ finite-horizon} & \multicolumn{1}{c|}{\makecell[c]{\Large$\frac{H^4SA}{\varepsilon^2}$ \\ \cite{li2020breaking}}} & \makecell[c]{\Large$\frac{H^3SA}{\varepsilon^2}$ \\ \cite{sidford2018near}} & \multicolumn{1}{c|}{\makecell[c]{\Large$\min\!\left\{\!\frac{H^{2.5}SA}{\varepsilon},\frac{H^3S\sqrt{A}}{\varepsilon}\right\}$ \\ \cite{luo2025quantum}}} & \multicolumn{1}{c|}{\cellcolor{blue!20}\mbox{\makecell[c]{\Large$\frac{H^{2.5}S\sqrt{A}}{\varepsilon}$ \\ \Cref{res:res1}}}} & \makecell[c]{\Large$\frac{H^{1.5}S\sqrt{A}}{\varepsilon}$ \\ \cite{luo2025quantum} } \\ \hline\hline

\!\makecell{Time-independent \\ infinite-horizon}\! & \multicolumn{1}{c|}{\makecell[c]{\Large$\frac{\Gamma^3SA}{\varepsilon^2}$  \\ \cite{sidford2018near} }} & \makecell[c]{\Large$\frac{\Gamma^3SA}{\varepsilon^2}$ \\ \cite{GheshlaghiAzar2013}} & \multicolumn{1}{c|}{\makecell[c]{\Large$\min\!\left\{\!\frac{\Gamma^{1.5}SA}{\varepsilon},\frac{\Gamma^3S\sqrt{A}}{\varepsilon}\right\}$ \\ \cite{wang2021quantum}}} & \multicolumn{1}{c|}{\cellcolor{blue!20}\mbox{\makecell[c]{\Large$\frac{\Gamma^{2.5}S\sqrt{A}}{\varepsilon}$ \\ \Cref{res:res2}}}} & \makecell[c]{\Large$\frac{\Gamma^{1.5}S\sqrt{A}}{\varepsilon}$ \\ \cite{wang2021quantum} } \\ \hline

\!\makecell{Time-dependent \\ infinite-horizon}\! & \multicolumn{1}{c|}{\makecell[c]{$-$}} & \makecell[c]{\Large$\frac{\Gamma^3SA}{\varepsilon^2}$ \\ \cite{GheshlaghiAzar2013}} & \multicolumn{1}{c|}{\makecell[c]{$-$}} & \multicolumn{1}{c|}{\cellcolor{blue!20}\mbox{\makecell[c]{\Large$\frac{\Gamma^{2.5}S\sqrt{A}}{\varepsilon}$ \\ \Cref{res:res2}}}} & \makecell[c]{\Large$\frac{\Gamma^{1.5}S\sqrt{A}}{\varepsilon}$ \\ \cite{wang2021quantum} } \\ \hline

\end{tabular}}
\label{table:results}
\end{table*}

\subsection{Review of Techniques}

It is known~\citep{puterman2014markov} that a solution to the equations below provides optimal value functions:
\begin{align*}
    V_H(s) &= 0 \quad\text{and}\quad V_h(s) = \max_{a\in\mathcal{A}}\Big\{r_h(s,a) + \sum_{s'\in\mathcal{S}} p_h(s'|s,a) V_{h+1}(s')\Big\} \quad \forall h\in[H] \quad\text{(finite-horizon)},\\
    V_0(s) &= 0 \quad\text{and}\quad V_{t+1}(s) = \max_{a\in\mathcal{A}}\Big\{r_t(s,a) + \gamma\sum_{s'\in\mathcal{S}} p_t(s'|s,a) V_{t}(s')\Big\} \quad \forall t\in\mathbb{N} \quad\text{(infinite-horizon)}.
\end{align*}
More precisely, for finite-horizon MDPs, $V_h = V_h^{\ast,H}$ for $h\in[H]$ is an optimal value function and 
\begin{align*}
    \pi = (\pi_h)_{h\in[H]} \quad\text{with} \quad \pi_{h}(s) \in \argmax_{a\in\mathcal{A}} \Big\{r_h(s,a) + \sum_{s'\in\mathcal{S}} p_h(s'|s,a) V_{h+1}(s')\Big\}
\end{align*}
is an optimal policy. For infinite-horizon discounted MDPs, the solution $V_t$ converges to the optimal value, i.e., $\lim_{t\to \infty} V_t = V_0^{\ast,\gamma}$. More usefully, for $t\in\mathbb{N}$ such that $\|V_{t+1} - V_t\|_\infty \leq \frac{\varepsilon}{2\gamma\Gamma}$, one have that $\|V_{t+1} - V_0^{\ast,\gamma}\|_\infty \leq \frac{\varepsilon}{2}$ and
\begin{align*}
    \pi=(\pi_0,\dots,\pi_t) \quad\text{with}\quad \pi_{t'+1} \in \argmax_{a\in\mathcal{A}} \Big\{r_{t'+1}(s,a) + \sum_{s'\in\mathcal{S}} p_{t'+1}(s'|s,a) V_{t'}(s')\Big\}
\end{align*}
is an $\varepsilon$-optimal policy~\cite[Theorem~6.3.1]{puterman2014markov}.

Given sampling access to $\{p_h\}_h$, the standard value iteration algorithm \citep{puterman2014markov} obtains an $\varepsilon$-optimal value function and policy for finite-horizon MDPs by approximating the quantities $\sum_{s'\in\mathcal{S}} p_h(s'|s,a) V_{h+1}(s')$ up to $\frac{\varepsilon}{H}$ precision for all $(s,a,h)\in\mathcal{S}\times\mathcal{A}\times[H-1]$, for a total classical query complexity of $O\big(\frac{H^5 SA}{\varepsilon^2}\big)$. For infinite-horizon discounted MDPs, the standard value iteration algorithm iteratively computes $V_{t}$ until $\|V_{t+1} - V_t\|_\infty \leq \frac{\varepsilon}{2\gamma\Gamma}$. Some contractive properties behind the problem imply that $O\big(\Gamma\log\frac{\Gamma}{\varepsilon}\big)$ iteration steps are sufficient. Hence, approximating the quantities $\sum_{s'\in\mathcal{S}} p_{t+1}(s'|s,a) V_{t}(s')$ up to additive error $\frac{\varepsilon}{2\Gamma}$ leads to an $\varepsilon$-optimal value function and $\Gamma\varepsilon$-optimal greedy policy (in the worst case) after $\widetilde{O}(\Gamma)$ iterations. The final classical query complexity is $\widetilde{O}\big(\frac{\Gamma^7 SA}{\varepsilon}\big)$.

In order to achieve the optimal query complexities $\widetilde{O}\big(\frac{H^3SA}{\varepsilon^2} \big)$ and $\widetilde{O}\big(\frac{\Gamma^3SA}{\varepsilon^2}\big)$, \cite{sidford2018near} employs three crucial techniques: \emph{monotonicity}, \emph{variance reduction}, and \emph{total-variance} techniques. The monotonicity technique means choosing an approximate value $V_t$ and deterministic decision rule $\pi_t:\mathcal{S}\to\mathcal{A}$ such that the monotonicity condition
\begin{align*}
	V_t(s) \leq r_t(s,\pi_t(s)) + \gamma \sum_{s'\in\mathcal{S}} p_t(s'|s,\pi_t(s)) V_t(s')
\end{align*}
is maintained throughout all time steps $t$ of the algorithm.\footnote{Finite-horizon MDPs and time-dependent infinite-horizon discounted MDPs do not require monotonicity since the optimal policy is non-stationary.} This is vital to guarantee that $V_t(s) \leq V_t^{\pi,\gamma}(s)$ and to obtain an $\varepsilon$-optimal policy from an $\varepsilon$-optimal function, otherwise an $\varepsilon$-optimal function only yields an $\Gamma\varepsilon$-optimal greedy policy in the worst case~\citep{bertsekas2022abstract}. The total-variance technique, on the other hand, is based on the fact that the true error accumulates as $\sqrt{\Gamma^3/m}$ given $m$ samples, much less than the naive sum of estimation errors at each time step. This means that one does not require an update error of $\frac{\epsilon}{2\Gamma}$ at each step and $m = O\big(\frac{\Gamma^3}{\epsilon^2}\big)$ samples suffices for a total error $\epsilon$, thus shaving off a factor of $\Gamma$. The precise statements for total-variance are given in \Cref{lem:upper_bound_variance_infinite_horizon}. Finally, we omit a description of the variance reduction technique since it will not be used in our quantum algorithms (variance reduction is actually what hinders the use of quantum max-finding in~\cite{wan2023quantum,luo2025quantum,ambainis2025bit}).

Our new algorithms are a quantised version of the standard value iteration but further improved with the monotonicity and total-variance techniques. For simplicity, we review the finite-horizon setting here. The infinite-horizon case is similar, with $\Gamma$ replacing $H$. At each time step $h$, we employ a quantum mean estimation subroutine to compute an approximation $\widetilde{\mu}_{h+1}(s,a) \approx \mu_{h+1}(s,a) := \sum_{s'\in\mathcal{S}} p_h(s'|s,a) V_{h+1}(s')$. Since we want to leverage the total-variance technique, we use the subroutine of Kothari and O'Donnell~\cite{Kothari2023mean} (\Cref{fact:quantum_mean_estimation_variance}) to approximate $\mu_{h+1}(s,a)$ up to error $O\big(\frac{\varepsilon\sqrt{\smash[b]{\sigma_{h+1}(s,a)}}}{H^{1.5}}\big)$ using its variance $\sigma_{h+1}(s,a)$. No prior knowledge of $\sigma_{h+1}(s,a)$ is required for the quantum subroutine. However, since we ultimately need an approximation with one-side error, we also approximate $\widetilde{\sigma}_{h+1}(s,a) \approx \sigma_{h+1}(s,a)$ separately up to error $O(\varepsilon\sqrt{H})$.

What is crucial to understand is that all the aforementioned approximations and calls to quantum mean estimation are done in \emph{superposition} over different values of $a\in\mathcal{A}$. This means that we treat these calculations as forming an oracle (for fixed $s\in\mathcal{S}$)
\begin{align*}
	|a\rangle |\bar{0}\rangle \mapsto |a\rangle |M_{h+1}(s,a)\rangle|\Sigma_{h+1}(s,a)\rangle,
\end{align*}
where $|M_{h+1}(s,a)\rangle$ and $|\Sigma_{h+1}(s,a)\rangle$ are quantum states that, upon measuring their first register in the computational basis, return good enough approximations $\widetilde{\mu}_{h+1}(s,a)\approx \mu_{h+1}(s,a)$ and $\widetilde{\sigma}_{h+1}(s,a) \approx \sigma_{h+1}(s,a)$ with high probability. By subtracting to content of $|\Sigma_{h+1}(s,a)\rangle$ from $|M_{h+1}(s,a)\rangle$ (so that the final approximation has an one-side error) and coherently adding $r_h(s,a)$ to the result, we can obtain the oracle
\begin{align*}
    \mathcal{U}_{s,h}: |a\rangle|\bar{0}\rangle \mapsto |a\rangle |\Phi_{h}(s,a)\rangle,
\end{align*}
where $|\Phi_h(s,a)\rangle$ is a quantum state that, upon measuring its first register, returns a good approximation to $r_h(s,a) + \sum_{s'\in\mathcal{S}} p_h(s'|s,a) V_{h+1}(s)$ with high probability $1-\delta(s,a)$. Note that $|\Phi_h(s,a)\rangle$ can hold several ``garbage'' registers that are ignored in the measurement process. The unitary $\mathcal{U}_{s,h}$ can be expressed as
\begin{align*}
    \mathcal{U}_{s,h}: |a\rangle|\bar{0}\rangle \mapsto |a\rangle(\sqrt{1-\delta(s,a)}|\Phi_{h}^{\rm good}(s,a)\rangle + \sqrt{\delta(s,a)}|\Phi_{h}^{\rm bad}(s,a)\rangle),
\end{align*}
where $|\Phi_{h}^{\rm good}(s,a)\rangle$ is the superposition over numbers that well approximate the quantity $r_h(s,a) + \sum_{s'\in\mathcal{S}} p_h(s'|s,a) V_{h+1}(s)$, while $|\Phi_{h}^{\rm bad}(s,a)\rangle$ is the superposition over numbers that do not.

At this point, we would like to input the oracle $\mathcal{U}_{s,h}$ onto a quantum max-finding algorithm so that $\max_{a\in\mathcal{A}}\{r_h(s,a) + \sum_{s'\in\mathcal{S}} p_h(s'|s,a) V_{h+1}(s)\}$ can be found. However, $\mathcal{U}_{s,h}$ yields a quantum state $|\Phi_{h}(s,a)\rangle$ that can be a superposition over several approximation values, which can break the basic algorithm of D\"urr and H\o{}yer~\cite{durr1996quantum}, an important detail that has been overlook by some previous works \citep{wang2021quantum,luo2025quantum}. In order to deal with this issue, we employ a more advanced quantum max-finding subroutine due to van Apeldoorn et al.~\cite{vanApeldoorn2020quantumsdpsolvers} (see \Cref{fact:quantum_minimum_finding}) that coherently filters a superposition of values by returning one which meets a certain threshold (usually selected as being close to the target maximum value). 
As an extra difficulty, $|\Phi_{h}(s,a)\rangle$ also includes states that are bad approximations (contained in $|\Phi_{h}^{\rm bad}(s,a)\rangle$). By following an argument due to Chen and de Wolf~\cite{chen2021quantum}, we show that this effectively reduces the overall success probability of the procedure by an additive term, which only leads to a polylogarithmic overhead in the query complexity.


\Cref{fact:quantum_minimum_finding} with $\widetilde{O}(\sqrt{A})$ calls to $\mathcal{U}_{s,h}$ returns $V_h(s) \approx \max_{a\in\mathcal{A}}\{r_h(s,a) + \sum_{s'\in\mathcal{S}} p_h(s'|s,a) V_{h+1}(s)\}$ with high probability and additive error $O\big(\frac{\varepsilon\sqrt{\smash[b]{\sigma_{h+1}(s,\pi_{h}(s))}}}{H^{1.5}}\big)$. 
After $H$ time steps, the total-variance technique guarantees that the error does not accumulate too rapidly and that function $V_0$ is $\varepsilon$-close to $V_0^{\ast,H}$. 
The final query complexity across all $s\in\mathcal{S}$ and time steps $H$ is $\widetilde{O}\big(\frac{H^{2.5}S\sqrt{A}}{\varepsilon}\big)$.

\paragraph*{Acknowledgments.} JFD is supported by the Lendület “Momentum” program of the Hungarian Academy of Sciences under grant agreement no.\ LP2025-8/2025.

\paragraph*{Artificial intelligence research statement.} Opus 5.5 has been used to check for typos and mistakes. Ideas, proofs, and writing are the product of the author alone.

\section{Preliminaries}

For $n\in\mathbb{N} := \{0,1,2,\dots\}$, let $[n]:=\{0, \dots, n-1\}$. Given a finite set $\mathcal{S}$, let $\mathscr{B}(\mathcal{S})$ be the space of all bounded Borel measurable real-valued functions on $\mathcal{S}$, which can be interpreted as $\mathbb{R}^S$ for $S = |\mathcal{S}|$, and let $\Delta(\mathcal{S})$ denote the probability simplex over $\mathcal{S}$. Let the $\ell_\infty$-norm of ${u}\in\mathscr{B}(\mathcal{S})$ be $\| u\|_\infty = \max_{s\in\mathcal{S}} |u(s)|$. We use $\mathbf{1}\in\mathscr{B}(\mathcal{S})$ to denote the all-ones function $\mathbf{1}(s) = 1$. Given $u,v\in\mathscr{B}(\mathcal{S})$, let $uv\in\mathscr{B}(\mathcal{S})$ be the function $(uv)(s) = u(s)v(s)$, and $u \leq v \iff u(s) \leq v(s)$.

\subsection{Quantum Computation}

Little background on quantum computation is needed for this paper and we refer to~\cite{nielsen2010quantum} for more information. The state of a quantum system is described by a unit vector from a Hilbert space denoted by the ket notation $|\cdot\rangle$. An $n$-qubit system is described by a unit vector in $\mathbb{C}^{2^n}$. The evolution of a quantum state $|\psi\rangle\in\mathbb{C}^{2^n}$ is described by a unitary operator $U\in\mathbb{C}^{2^n\times 2^n}$, $UU^\dagger = I$ where $U^\dagger$ is the Hermitian conjugate of $U$. In order to extract classical information from a quantum system, a quantum measurement is performed, which is a set $\{E_m\}_m$ of positive operators $E_m \succ 0$ that sum to identity, $\sum_m E_m = I$. The probability of measuring $E_m$ on $|\psi\rangle$ is $\langle\psi|E_m|\psi\rangle$. 
We use $|\bar{0}\rangle$ to denote the state $|0\rangle\otimes\cdots\otimes|0\rangle$ for a sufficient number of qubits.

In this paper, we employ quantum oracles for functions and probability distributions. We say we have quantum access to a function $u\in\mathscr{B}(\mathcal{S})$ if we have access to the unitary
\begin{align*}
    \mathcal{O}_u:|s\rangle|\bar{0}\rangle \mapsto |s\rangle|u(s)\rangle|\mathrm{garbage}(s)\rangle \qquad \forall s\in\mathcal{S}
\end{align*}
and its inverse, and we say we have quantum sampling access to a probability distribution $p\in\Delta(\mathcal{S})$ if we have access to the unitary 
\begin{align*}
    \mathcal{O}_p:|\bar{0}\rangle \mapsto \sum_{s\in\mathcal{S}} \sqrt{p(s)}|s\rangle|\mathrm{garbage}(s)\rangle \qquad \forall s\in\mathcal{S}
\end{align*}
 and its inverse, where $|\mathrm{garbage}(s)\rangle$ are normalised ``garbage'' states. Quantum access to a function is usually referred to as a quantum random access memory (QRAM)~\citep{giovannetti2008architectures,giovannetti2008quantum,jaques2023qram,allcock2023constant}. Quantum access to $u\in\mathscr{B}(\mathcal{S})$ can be built in $\widetilde{O}(S)$ time. We shall require a quantum mean estimation subroutine due to Kothari and O'Donnell~\cite{Kothari2023mean} and a generalised version of quantum max-finding due to van Apeldoorn et al.~\cite{vanApeldoorn2020quantumsdpsolvers}.

\begin{fact}[Quantum mean estimation with variance {\cite[Theorem~1.1]{Kothari2023mean}}]\label{fact:quantum_mean_estimation_variance}
    Let $\epsilon>0$ and $\delta\in (0, 1)$. Assume quantum access to $u:\mathcal{S}\to \mathbb{R}$ via oracle $\mathcal{O}_u$ and quantum sampling access to $p\in\Delta(\mathcal{S})$ via oracle $\mathcal{O}_p$. Let $\sigma := \sum_{s\in\mathcal{S}} p(s) u(s)^2 - \big(\sum_{s\in\mathcal{S}} p(s) u(s)\big)^2$. There is a quantum algorithm that computes $\widetilde \mu\in\mathbb{R}$ such that $| \widetilde\mu - \sum_{s\in\mathcal{S}} p(s) u(s)|\leq \sqrt{\sigma}\epsilon$ with success probability $1-\delta$ using $O\big(\frac{1}{\epsilon}\log\frac{1}{\delta}\big)$ queries to $\mathcal{O}_u,\mathcal{O}_p$, and their inverses.
\end{fact}
%
\begin{fact}[Generalised quantum maximum finding {\cite[Theorem~49]{vanApeldoorn2020quantumsdpsolvers}}]\label{fact:quantum_minimum_finding}
    Let $m\in\mathbb{R}$ and $\delta\in(0,1)$. Assume access to a unitary $U:|\bar{0}\rangle \mapsto \sum_{i=1}^N \sqrt{p_i}|\psi_i\rangle|x_i\rangle$, where $|\psi_i\rangle$ are normalised states and $x_1 < x_2 < \cdots < x_N$ are real numbers. Define the random variable $X$ according to $\mathrm{Pr}[X = x_i] = p_i$. Let $K = O\Big(\frac{\log(1/\delta)}{\sqrt{\mathrm{Pr}[X \geq m]}}\Big)$ be a natural number. There is a quantum algorithm that obtains a state $|\psi_i\rangle|x_i\rangle$ where $x_i\geq m$ with probability at least $1-\delta$ using $K$ applications of $U$ and $U^\dagger$.
\end{fact}

\subsection{Background on Markov Decision Processes}

A discrete-time MDP models the interaction between an agent and an environment. At each time step $t\in\mathbb{N}$, the environment is in some state $s_t$ from a finite set $\mathcal{S}$ of size $S$ and the agent chooses an action $a_t$ from a finite set $\mathcal{A}$ of size $A$ according to a \emph{randomised decision rule} $\pi_t : \mathcal{S}\to\Delta(\mathcal{A})$ so that $a_t \sim \pi_t(\cdot|s_t)$. If $\pi_t : \mathcal{S}\to\mathcal{A}$, we say $\pi_t$ is \emph{deterministic}. After choosing action $a_t$, the agent receives a reward $r_t(s_t,a_t)$ given by the measurable bounded reward function $r_t:\mathcal{S}\times\mathcal{A}\to[0,1]$ and the environment transitions to a new state $s_{t+1}\sim p_t(\cdot|s_t,a_t)$ according to the stochastic kernel $p_t:\mathcal{S}\times\mathcal{A}\to\Delta(\mathcal{S})$. A \emph{randomised policy} $\pi = (\pi_t)_t$ is a sequence of decision rules $\pi_t:\mathcal{S}\to\Delta(\mathcal{A})$, and similarly for a \emph{deterministic policy}. A policy is \emph{stationary} if $\pi_t = \pi_{t'}$ for all $t,t'$. Given $d\in\mathcal{D}^{\rm R}$, we use the notation $d^\infty$ for the stationary policy $(d,d,\dots)$. Let $\mathcal{D}^{\rm R}$ and $\mathcal{D}^{\rm D}$ be the set of all randomised and deterministic decision rules and $\Pi^{\rm R}$ and $\Pi^{\rm D}$ the set of all randomised and deterministic policies, respectively. Given $d\in\mathcal{D}^{\rm R}$, define $r^d_t\in\mathscr{B}(\mathcal{S})$ and $p_t^{d}:\mathcal{S}\to\Delta(\mathcal{S})$ as
\begin{align*}
    r_t^{d}(s) := \operatorname*{\mathbb{E}}_{a\sim d(\cdot|s)}[r_t(s,a)] \qquad\text{and} \qquad p_t^{d}(s'|s) := \operatorname*{\mathbb{E}}_{a\sim d(\cdot|s)}[p_t(s'|s,a)].
\end{align*}  
Given stochastic kernel $p_t:\mathcal{S}\times\mathcal{A}\to\Delta(\mathcal{S})$, time step $t\in\mathbb{N}$, decision rule $d\in\mathcal{D}^{\rm R}$, and action $a\in\mathcal{A}$, let the operators $\mathcal{P}_t^d, \mathcal{P}_t^a:\mathscr{B}(\mathcal{S}) \to \mathscr{B}(\mathcal{S})$ be defined as 
\begin{align*}
    (\mathcal{P}_t^d u)(s) = \sum_{s'\in\mathcal{S}} p^d_t(s'|s) u(s') \quad\text{and}\quad (\mathcal{P}_t^a u)(s) = \sum_{s'\in\mathcal{S}} p_t(s'|s,a)u(s') \quad \forall u\in\mathscr{B}(\mathcal{S}), s\in\mathcal{S}. 
\end{align*}

\paragraph{Finite-horizon MDPs.} Central to this work are \emph{finite-horizon} MDPs where the agent interacts with the environment for a pre-determined number of steps $H\in\mathbb{N}$ called \emph{horizon}. A time-dependent finite-horizon MDP is described by the tuple $\langle \mathcal{S},\mathcal{A},\{p_h\}_{h\in[H]},\{r_h\}_{h\in[H]},H\rangle$. In the time-independent setting, $r_h=r_{h'}$ and $p_h = p_{h'}$ for all $h,h'$. Given policy $\pi = (\pi_h)_{h\in[H]}$ and $h\in[H]$, let the \emph{value function} $V_h^{\pi,H} : \mathcal{S} \to [0,H]$ be defined as
\begin{align*}
    V_h^{\pi,H}(s) = \mathbb{E}\left[\sum_{h'=h}^{H-1} r_{h'}(s_{h'},a_{h'}) ~\Bigg|~ s_h = s, a_{h'} \sim \pi_{h'}(s_{h'}), s_{{h'}+1} \sim p_{h'}(\cdot|s_{h'}, a_{h'})\right] ~\forall s\in\mathcal{S}.
\end{align*}
Not all decision rules that make up $\pi$ are necessarily employed in $V_h^{\pi,H}$. A policy 
\begin{align*}
    \pi_\varepsilon\in\Pi^{\rm R} \quad\text{is}~ \varepsilon\text{-\emph{optimal} for}~ \varepsilon \geq 0 \quad\text{if} \quad V^{\pi_\varepsilon,H}_0 \geq V^{\pi,H}_0 - \varepsilon\mathbf{1} \quad\text{for all}~ \pi\in\Pi^{\rm R}
\end{align*}
and a policy $\pi^\ast$ is \emph{optimal} if it is $(\varepsilon=0)$-optimal. It is well known that the optimal policy is deterministic. The \emph{optimal value function} is
\begin{align*}
    V^{\ast,H}_h := \sup_{\pi\in\Pi^{\rm R}} V_h^{\pi,H} = \sup_{\pi\in\Pi^{\rm D}} V_h^{\pi,H}.
\end{align*}

Let the \emph{$h$-step Bellman operator} $\mathcal{L}_h^d:\mathscr{B}(\mathcal{S})\to \mathscr{B}(\mathcal{S})$ associated with $d\in\mathcal{D}^{\rm R}$ be 
\begin{align*}
    \mathcal{L}_h^d u = r_h^d + \mathcal{P}^d_h u, \qquad\text{or less compactly}, \qquad (\mathcal{L}_d u)(s) = r_h^d(s) + \sum_{s'\in\mathcal{S}} p^d_h(s'|s) u(s'),
\end{align*}
%
%
and the \emph{$h$-step optimal Bellman operator} $\mathcal{L}_h:\mathscr{B}(\mathcal{S})\to \mathscr{B}(\mathcal{S})$ be 
\begin{align*}
    \mathcal{L}_h u = \sup_{d\in\mathcal{D}^{\rm R}} \mathcal{L}_h^d u, ~\text{or less compactly}, ~ (\mathcal{L}_h u)(s) = \max_{a\in\mathcal{A}} \bigg\{r_h(s,a) + \sum_{s'\in\mathcal{S}} p_h(s'|s,a) u(s') \bigg\}.
\end{align*}
%
%
The Bellman operators $\mathcal{L}_h$ and $\mathcal{L}_h^d$ are monotonic, i.e., $u\leq v \implies \mathcal{L}_hu \leq \mathcal{L}_h v$, and non-expansive, i.e., $\|\mathcal{L}_hu - \mathcal{L}_hv\|_\infty \leq \|u-v\|_\infty$ for all $u,v\in\mathscr{B}(\mathcal{S})$ (similarly for $\mathcal{L}_h^d$). Finally, the value functions $V_h^{\pi,H}$ satisfy the recursive equation 
\begin{align*}
    V_h^{\pi,H} = \mathcal{L}_h^{\pi_h} V_{h+1}^{\pi,H} = r_h^{\pi_h} + \mathcal{P}_h^{\pi_h}V_{h+1}^{\pi,H}.
\end{align*}

\paragraph{Infinite-horizon discounted MDPs.} Another important class of MDPs are \emph{infinite-horizon discounted MDPs} where the agent-environment interaction lasts an infinite amount of time but at every time step the received reward is decreased by a \emph{discount factor} $\gamma\in[0,1)$. The parameter $\Gamma := (1-\gamma)^{-1}$ is usually called \emph{effective horizon}. A time-dependent infinite-horizon discounted MDP is described by the tuple $\langle \mathcal{S},\mathcal{A},\{p_t\}_{t\in\mathbb{N}},\{r_t\}_{t\in\mathbb{N}},\gamma\rangle$. In the time-independent setting, $r_t=r_{t'}$ and $p_t = p_{t'}$ for all $t,t'$. Given policy $\pi = (\pi_t)_{t\in\mathbb{N}}$ and $t,T\in\mathbb{N}$ with $t< T$, let the \emph{value function} $V_t^{\pi,T}:\mathcal{S}\to [0,\Gamma]$ be defined as
\begin{align*}
    V_t^{\pi,T}(s) = \mathbb{E}\left[\sum_{t'=t}^{T-1} \gamma^{t'-t} r_{t'}(s_{t'},a_{t'}) ~\Bigg|~ s_t = s, a_{t'} \sim \pi_{t'}(s_{t'}), s_{t'+1} \sim p_{t'}(\cdot|s_{t'}, a_{t'})\right] ~\forall s\in\mathcal{S}.
\end{align*}
Define also the value function as $T\to\infty$,
\begin{align*}
    V_t^{\pi,\gamma}(s) = \mathbb{E}\left[\sum_{t'=t}^\infty \gamma^{t'-t} r_{t'}(s_{t'},a_{t'}) ~\Bigg|~ s_t = s, a_{t'} \sim \pi_{t'}(s_{t'}), s_{t'+1} \sim p_{t'}(\cdot|s_{t'}, a_{t'})\right] ~\forall s\in\mathcal{S}.
\end{align*}
Not all decision rules that make up $\pi$ are necessarily employed in $V_t^{\pi,T}, V_t^{\pi,\gamma}$. A policy 
\begin{align*}
    \pi_\varepsilon\in\Pi^{\rm R} \quad\text{is}~ \varepsilon\text{-\emph{optimal} for}~ \varepsilon \geq 0 \quad\text{if} \quad V^{\pi_\varepsilon,\gamma}_0 \geq V^{\pi,\gamma}_0 - \varepsilon\mathbf{1} \quad\text{for all}~ \pi\in\Pi^{\rm R}
\end{align*}
and a policy $\pi^\ast$ is \emph{optimal} if it is $(\varepsilon=0)$-optimal. Similar definitions hold for $V_0^{\pi,T}$. It is well known that the optimal policy is deterministic. The \emph{optimal value functions} are
\begin{align*}
    V_t^{\ast,\gamma} := \sup_{\Pi^{\rm R}} V_t^{\pi, \gamma} = \sup_{\Pi^{\rm D}} V_t^{\pi, \gamma} \quad\text{and}\quad V_t^{\ast,T} := \sup_{\Pi^{\rm R}} V_t^{\pi,T} = \sup_{\Pi^{\rm D}} V_t^{\pi,T}.
\end{align*}
%

Let the \emph{$t$-step Bellman operator} $\mathcal{L}^d_t:\mathscr{B}(\mathcal{S})\to \mathscr{B}(\mathcal{S})$ associated with $d\in\mathcal{D}^{\rm R}$ and the \emph{$t$-step  optimal Bellman operator} $\mathcal{L}_t:\mathscr{B}(\mathcal{S})\to \mathscr{B}(\mathcal{S})$ be defined as 
\begin{align*}
    \mathcal{L}^d_t u = r^d_t + \gamma \mathcal{P}^d_t u \qquad\text{and}\qquad \mathcal{L}_t u = \sup_{d\in\mathcal{D}^{\rm R}} \mathcal{L}^d_t u.
\end{align*}
The Bellman operators $\mathcal{L}_t$ and $\mathcal{L}^d_t$ are monotonic i.e., $u\leq v \implies \mathcal{L}_tu \leq \mathcal{L}_tv$, and contractive, i.e., $\|\mathcal{L}_tu - \mathcal{L}_tv\|_\infty \leq \gamma\|u-v\|_\infty$ for all $u,v\in\mathscr{B}(\mathcal{S})$ (similarly for $\mathcal{L}_t^d$). Similarly to the finite-horizon setting, the value functions $V_t^{\pi,T}$ and $V_t^{\pi,\gamma}$ satisfy the recursive equations 
\begin{align*}
    V_{t}^{\pi,T} = \mathcal{L}^{\pi_{t}}_{t}V_{t+1}^{\pi,T} \qquad\text{and}\qquad V_{t}^{\pi,\gamma} = \mathcal{L}^{\pi_{t}}_{t}V_{t+1}^{\pi,\gamma}.
\end{align*}
In the \emph{time-independent} setting, the value functions $V_t^{\pi,\gamma}$ all collapse, i.e., $V_t^{\pi,\gamma} = V_{t'}^{\pi,\gamma}$ for any $t,t'\in\mathbb{N}$. Moreover, given $d\in\mathcal{D}^{\rm R}$, the value function $V_0^{d^\infty,\gamma}$ is the unique fixed point of the Bellman operator $\mathcal{L}^d$, i.e., $\mathcal{L}^d V_0^{d^\infty,\gamma} = V_\infty^{d^\infty,\gamma}$, while the optimal value function $V_0^{\ast,\gamma}$ is the unique fixed point of the Bellman operator $\mathcal{L}$, i.e., $\mathcal{L}V_0^{\ast,\gamma} = V_0^{\ast,\gamma}$. As a consequence, for any $d\in\mathcal{D}^{\rm R}$ and $u\in\mathscr{B}(\mathcal{S})$, $(\mathcal{L}^d)^\infty u = V_0^{d^\infty,\gamma}$ and $\mathcal{L}^\infty u = V_0^{\ast,\gamma}$.

We shall need the following technical result (total-variance technique) showing that the accumulated variance of all value functions does not grow too fast. It generalises~\cite[Lemma~F.4]{sidford2018near} and applies to both finite- and infinite-horizon discounted MDPs.

\begin{lemma}[Total-variance]\label{lem:upper_bound_variance_infinite_horizon}
    Let $\gamma\in[0,1]$ and $T\in\mathbb{N}$. Given policy $\pi = (\pi_t)_{t\in[T]}$, let $\sigma^\pi_{t}\in\mathscr{B}(\mathcal{S})$ be $\sigma_{t+1}^\pi = \mathcal{P}_t^{\pi_{t}} (V_{t+1}^{\pi,T})^2 - (\mathcal{P}_t^{\pi_{t}} V_{t+1}^{\pi,T})^2$ for $t\in[T]$. For any policy $\pi = (\pi_t)_{t\in[T]}$ and $t\in[T]$,
    \begin{align*}
        \left\|\sum_{t'=t}^{T-2} \left(\prod_{i=t}^{t'-1} \gamma \mathcal{P}_i^{\pi_i}\right) \sqrt{\gamma\sigma_{t'+1}^\pi}\right\|_\infty \leq \min\{\Gamma^{3/2}, T^{3/2}\}.
    \end{align*}
\end{lemma}
\begin{proof}
    Observe that
    \begin{align*}
        \gamma\sigma_{t+1}^\pi &\leq \gamma \mathcal{P}_t^{\pi_{t}} (V_{t+1}^{\pi,T})^2 - (\gamma \mathcal{P}_t^{\pi_{t}} V_{t+1}^{\pi,T})^2 \\
        &= \gamma \mathcal{P}_t^{\pi_{t}} (V_{t+1}^{\pi,T})^2 - (V_{t}^{\pi,T} - r_t^{\pi_{t}})^2 \tag{$V_{t}^{\pi,T} = r_t^{\pi_{t}} + \gamma \mathcal{P}_t^{\pi_{t}}V_{t+1}^{\pi,T}$} \\
        &= \gamma \mathcal{P}_t^{\pi_{t}} (V_{t+1}^{\pi,T})^2 - (V_{t}^{\pi,T})^2 + 2r_t^{\pi_t} V_{t}^{\pi,T} - (r_t^{\pi_t})^2.
    \end{align*}
    Write for simplicity $R_t = r_t^{\pi_t}$ and $P_{t'} = \prod_{i=t}^{t'-1}\gamma \mathcal{P}_i^{\pi_i}$, so that $P_{t'}\gamma\mathcal{P}_{t'}^{\pi_{t'}} = P_{t'+1}$. Therefore,
    \begin{align*}
        \sum_{t'=t}^{T-2} P_{t'} \gamma\sigma_{t'+1}^\pi &\leq \sum_{t'=t}^{T-2} P_{t'}\!\left(\gamma \mathcal{P}_{t'}^{\pi_{t'}} (V_{t'+1}^{\pi,T})^2 - (V_{t'}^{\pi,T})^2 + 2R_{t'} V_{t'}^{\pi,T} - R_{t'}^2 \right)\\
        &= P_{T-1}(V_{T-1}^{\pi, T})^2 - (V_t^{\pi, T})^2 + \sum_{t'=t}^{T-2} P_{t'}\big(2R_{t'} V_{t'}^{\pi,T} - R_{t'}^2 \big) \tag{telescope sum} \\
        &= \sum_{t'=t}^{T-1} P_{t'}\big(2R_{t'} V_{t'}^{\pi,T} - R_{t'}^2 \big) - (V_t^{\pi, T})^2 \tag{by $(V_{T-1}^{\pi, T})^2 = 2R_{T-1} V_{T-1}^{\pi, T} - R_{T-1}^2$ since $V_{T-1}^{\pi, T} = R_{T-1}$} \\
        &\leq 2\min\{\Gamma,T\}\sum_{t'=t}^{T-1} P_{t'}R_{t'} - (V_t^{\pi, T})^2 \tag{$R^2_{t'} \geq 0$ and $V_{t'}^{\pi,T}\leq \min\{\Gamma,T\}$} \\
        &= 2\min\{\Gamma,T\} V_t^{\pi,T} - (V_t^{\pi, T})^2 \tag{$\sum_{t'=t} P_{t'}R_{t'} = V_t^{\pi, T}$} \\
        &= \min\{\Gamma,T\}^2 - (\min\{\Gamma,T\} - V_{t}^{\pi, T})^2\\
        &\leq \min\{\Gamma^2,T^2\}. \tag{$V_t^{\pi, T} \leq \min\{\Gamma,T\}$}
    \end{align*}
    The final result follows from Cauchy-Schwarz inequality, i.e.,
    \[
        \left\|\sum_{t'=t}^{T-2} \left(\prod_{i=t}^{t'-1} \gamma \mathcal{P}_i^{\pi_i}\right) \sqrt{\gamma\sigma_{t'+1}^\pi}\right\|_\infty \leq \left\|\min\{\Gamma,T\}\sum_{t'=t}^{T-2} \left(\prod_{i=t}^{t'-1} \gamma \mathcal{P}_i^{\pi_i}\right) \gamma\sigma_{t'+1}^\pi\right\|_\infty^{1/2} \leq \min\{\Gamma^{3/2}, T^{3/2}\}. \qedhere
    \]
\end{proof}

\section{Optimal Policies for Finite-Horizon MDPs}
\label{sec:finite_horizon}

In this section, we formally describe and analyse our quantum algorithm (\Cref{algo:finite_horizon}) for computing $\varepsilon$-optimal policies for finite-horizon MDPs. 

\begin{theorem}\label{thr:quantum_finite-horizon2}
    Let $\langle \mathcal{S},\mathcal{A}, \{p_h\}_{h\in[H]}, \{r_h\}_{h\in[H]},H\rangle$ be a time-dependent finite-horizon MDP. Let $\delta\in(0,1)$ and $\varepsilon\in(0,\sqrt{H}]$. {\rm \Cref{algo:finite_horizon}} computes functions $\{V_h\}_{h\in[H]}\subset\mathscr{B}(\mathcal{S})$ and deterministic policy $\pi=(\pi_h)_{h\in[H]}$ such that, with probability $1-\delta$,
    \begin{align*}
        V^{\ast,H}_h - \varepsilon\mathbf{1} \leq V_h \leq V^{\pi,H}_h \leq  V^{\ast,H}_h \qquad\forall h\in[H].
    \end{align*}
    Its query complexity is (up to $\poly\log\log$ factors)
    \begin{align*}
        \widetilde{O}\bigg(\frac{H^{2.5}S\sqrt{A}}{\varepsilon}\log\left(\frac{HSA}{\delta} \right)\log\left(\frac{HS}{\delta} \right) \bigg).
    \end{align*}
\end{theorem}
\begin{algorithm}[t!]
    \caption{Quantum algorithm for optimal policies for finite-horizon MDPs}
    \label{algo:finite_horizon}
    \begin{algorithmic}[1]  
    \Require Finite state space $\mathcal{S}$ and action space $\mathcal{A}$, horizon $H$, quantum sampling access to $\{p_h\}_{h\in[H]}$, quantum access to $\{r_h\}_{h\in[H]}$, failure probability $\delta\in(0, 1)$, error $\varepsilon \in (0,\sqrt{H}]$.

    \Ensure $\varepsilon$-optimal deterministic policy $\pi$.

    \State $\theta \gets \frac{\varepsilon}{5H^{3/2}}$, $V_{H-1}(s) \gets \max_{a\in\mathcal{A}}r_{H-1}(s,a)$, and $\pi_{H-1}(s) \gets \argmax_{a\in\mathcal{A}}r_{H-1}(s,a)$
        
        \For{$h=H-2,H-3,\dots,0$}

            \State Let $\mu_{h+1}(s,a) = \sum_{s'\in\mathcal{S}}p_h(s'|s,a) V_{h+1}(s')$ 
            
            \State Let $\sigma_{h+1}(s,a) = \sum_{s'\in\mathcal{S}}p_h(s'|s,a) V_{h+1}(s')^2 - \big(\sum_{s'\in\mathcal{S}}p_h(s'|s,a) V_{h+1}(s')\big)^2$ 

            \For{$s\in\mathcal{S}$}

            \State \parbox[t]{\dimexpr\linewidth-\algorithmicindent-\algorithmicindent}{Use \Cref{fact:quantum_mean_estimation_variance} to obtain a unitary $\mathcal{U}_{s}^{(h)}:|a\rangle|\bar{0}\rangle \mapsto |a\rangle|\Phi_{h}(s,a)\rangle$, where $|\Phi_{h}(s,a)\rangle$ is a quantum state that, upon measuring its first register in the computational basis, returns
            \begin{align*}
                Q_{h}(s,a) = r_h(s,a) + \widetilde{\mu}_{h+1}(s,a) - \theta\sqrt{\widetilde{\sigma}_{h+1}(s,a)} - \theta^{3/2} H
            \end{align*}
            for $\widetilde{\mu}_{h+1}(s, a)$ and $\widetilde{\sigma}_{h+1}(s, a)$ that satisfy, with failure probability $\widetilde{O}\big(\frac{\delta^2}{H^2 S^2 A}\big)$,
            \begin{align*}
                |\widetilde{\mu}_{h+1}(s,a) - \mu_{h+1}(s,a) | &\leq \theta\sqrt{\sigma_{h+1}(s,a)},\\
                |\widetilde{\sigma}_{h+1}(s,a) - \sigma_{h+1}(s,a) | &\leq \theta H^2
            \end{align*}}

            \State \parbox[t]{\dimexpr\linewidth-\algorithmicindent-\algorithmicindent}{Use \Cref{fact:quantum_minimum_finding} with $W_s^{(h)} = \mathcal{U}_s^{(h)}(\mathsf{H}^{\lceil\log_2{A}\rceil}\otimes I)$ (where $\mathsf{H}$ is the Hadamard gate) to obtain $V_h(s)$ and $\pi_h(s)$ such that, with probability $1-\frac{\delta}{HS}$,
			\begin{align*}
				V_h(s) &\leq r_h(s,\pi_h(s)) + \mu_{h+1}(s,\pi_h(s)),\\
				V_h(s) &\geq r_h(s,\pi_h(s)) + \mu_{h+1}(s,\pi_h(s)) - 2\theta\sqrt{\sigma_{h+1}(s,\pi_h(s))} - 2\theta^{3/2}H
			\end{align*}}
		
            \EndFor

        \EndFor
    
    \State {\bfseries return} $\pi = (\pi_0,\dots,\pi_{H-1})$
\end{algorithmic}
\end{algorithm}
\begin{proof}
    The proof is by (backwards) induction on $h\in[H]$, with the base case $h=H$ being trivial since $V_{H} \equiv 0$. Assume then that (the inequality is entry-wise)
    \begin{align*}
        V^{\ast,H}_{h'} - \varepsilon\mathbf{1} \leq V_{h'} \leq V_{h'}^{\pi,H} \leq V^{\ast,H}_{h'} \qquad\forall h'=h+1,h+2,\dots,H.
    \end{align*}
    The main idea is to employ \Cref{fact:quantum_mean_estimation_variance} to compute approximations for the true quantities
    \begin{align*}
        \mu_{h+1}(s,a) := (\mathcal{P}_h^a V_{h+1})(s) \qquad\text{and}\qquad  \sigma_{h+1}(s,a) := (\mathcal{P}_h^a V_{h+1}^2)(s) - (\mathcal{P}_h^a V_{h+1})^2(s).
    \end{align*}
    Let $\theta := \frac{\varepsilon}{5 H^{3/2}}$. We claim that \Cref{fact:quantum_mean_estimation_variance} can be adapted to perform all of its steps in superposition without any need for intermediary measurements. This means that we can invoke \Cref{fact:quantum_mean_estimation_variance} twice to construct a unitary
    \begin{align}\label{eq:superposition_approximate_quantities}
        \mathcal{O}_s^{(h)} : |a\rangle |\bar{0}\rangle \mapsto |a\rangle|M_{h+1}(s,a)\rangle|\Sigma_{h+1}(s,a)\rangle,
    \end{align}
    where $|M_{h+1}(s,a)\rangle$ and $|\Sigma_{h+1}(s,a)\rangle$ are states that, upon measuring their first register in the computational basis, return $\widetilde{\mu}_{h+1}(s,a)$ and $\widetilde{\sigma}_{h+1}(s,a)$ such that, with probability $1-\delta_2(s,a)$,
    \begin{align}\label{eq:approximations}
        \forall(s,a)\in\mathcal{S}\times\mathcal{A} : \qquad
        \begin{aligned}
            &|\widetilde{\mu}_{h+1}(s, a) - \mu_{h+1}(s,a) | \leq \theta\sqrt{\sigma_{h+1}(s,a)},\\
            &| \widetilde{\sigma}_{h+1}(s, a) - \sigma_{h+1}(s,a) | \leq \theta H^2,
        \end{aligned}
    \end{align}
    using that the standard deviation of $V^2_{h+1}$ is at most $H^2$ and thus one can compute approximations of $(\mathcal{P}_a V_{h+1}^2)(s)$ and $(\mathcal{P}_a V_{h+1})^2(s)$ with error at most $\theta H^2/2$ each. 
    
    By subtracting the content of register $|\Sigma_{h+1}(s,a)\rangle$ from the register $|M_{h+1}(s,a)\rangle$ in \eqref{eq:superposition_approximate_quantities} and coherently adding $r_h(s,a) - \theta^{3/2}H$ to $|M_{h+1}(s,a)\rangle$, it is then possible to effectively construct a black-box oracle $\mathcal{U}_s^{(h)}:|a\rangle|\bar{0}\rangle \mapsto |a\rangle|\Phi_{h}(s,a)\rangle$ using oracles $\mathcal{Q}_{p_h},\mathcal{Q}_{p_h}^\dagger$, where $|\Phi_{h}(s,a)\rangle$ is a quantum state that, upon measuring its first register in the computational basis, returns
    \begin{align*}
        \forall(s,a)\in\mathcal{S}\times\mathcal{A}: \qquad  Q_{h}(s,a) = r_h(s,a) + \widetilde{\mu}_{h+1}(s,a) - \theta\sqrt{\widetilde{\sigma}_{h+1}(s,a)} - \theta^{3/2} H,
    \end{align*}
    where $\widetilde{\mu}_{h+1}(s, a)$ and $\widetilde{\sigma}_{h+1}(s, a)$ are quantities that satisfy \eqref{eq:approximations} with probability $1-\delta_2(s,a)$. 
    
    Now note that, using \eqref{eq:approximations},
    \begin{align*}
        \forall(s,a)\in\mathcal{S}\times\mathcal{A}: \qquad 
        \begin{aligned}
            \widetilde{\mu}_{h+1}(s, a) &\leq \mu_{h+1}(s,a) + \theta\sqrt{\widetilde{\sigma}_{h+1}(s,a)} + \theta^{3/2}H,\\
            \widetilde{\mu}_{h+1}(s, a) &\geq \mu_{h+1}(s,a) - \theta\sqrt{\sigma_{h+1}(s,a)}.
        \end{aligned}
    \end{align*}
    On the other hand, the optimal-policy variance $\sigma_{h+1}^\ast(s,a) := (\mathcal{P}_h^a (V_{h+1}^{\ast,H})^2)(s) - (\mathcal{P}_h^a V_{h+1}^{\ast,H})^2(s)$ can be used to bound
    \begin{align*}
        \forall (s,a)\in\mathcal{S}\times\mathcal{A}: \qquad \sqrt{\widetilde{\sigma}_{h+1}(s,a)} \leq \sqrt{\sigma_{h+1}(s,a)} + \sqrt{\theta}H \leq \sqrt{\sigma^\ast_{h+1}(s,a)} + \varepsilon + \sqrt{\theta}H,
    \end{align*}
    using \eqref{eq:approximations} and $\sqrt{\operatorname{Var}[V_{h+1}]} \leq \sqrt{\operatorname{Var}[V^{\ast,H}_{h+1}]} + \sqrt{\operatorname{Var}[V^{\ast,H}_{h+1} - V_{h+1}]} \leq \sqrt{\operatorname{Var}[V^{\ast,H}_{h+1}]} + \varepsilon$ if $\|V^{\ast,H}_{h+1} - V_{h+1}\|_\infty \leq \varepsilon$ according to the induction hypothesis. This means that measuring $|\Phi_h(s,a)\rangle$ in the computational basis returns $Q_h(s,a)$ such that, with probability $1-\delta_2(s,a)$,
    \begin{align}\label{eq:approximations_final}
        \forall(s,a)\in\mathcal{S}\times\mathcal{A} : \qquad
        \begin{aligned}
            &Q_{h}(s, a) \leq r_h(s,a) + \mu_{h+1}(s,a), \\
            &Q_{h}(s, a) \geq r_h(s,a) + \mu_{h+1}(s,a) - \varepsilon_{h}(s,a),
        \end{aligned}
    \end{align}
    where $\varepsilon_{h}(s,a) := 2\theta\sqrt{\sigma^\ast_{h+1}(s,a)} + 2\theta\varepsilon + 2\theta^{3/2} H$.

    
    Without loss of generality, assume $\log_2{A}$ is a natural number. Write
    \begin{align*}
        |\Phi_{h}(s,a)\rangle = \sqrt{1-\delta_2(s,a)}|\Phi_{h}^{\rm good}(s,a)\rangle + \sqrt{\delta_2(s,a)}|\Phi_{h}^{\rm bad}(s,a)\rangle,
    \end{align*}
    where $|\Phi_{h}^{\rm good}(s,a)\rangle$ is the superposition over numbers that satisfy \eqref{eq:approximations_final}, while $|\Phi_{h}^{\rm bad}(s,a)\rangle$ is the superposition over numbers that do not satisfy \eqref{eq:approximations_final}. Assume for now access to a unitary $\mathcal{U}_{s,\text{ideal}}^{(h)}:|a\rangle|\bar{0}\rangle \mapsto |a\rangle|\Phi_{h}^{\rm good}(s,a)\rangle$ and let $W_s^{(h)} = \mathcal{U}_{s,\text{ideal}}^{(h)}(\mathsf{H}^{\otimes \log_2{A}}\otimes I)$, where $\mathsf{H}$ is the Hadamard gate. Let
    \begin{align*}
        V_h^{\rm aux}(s) &:= \max_{a\in\mathcal{A}}\{r_h(s,a) + \mu_{h+1}(s,a) - \varepsilon_{h}(s,a)\} \\
        &\geq r_h(s,\pi_h^\ast(s)) + \mu_{h+1}(s,\pi_h^\ast(s)) - \varepsilon_{h}(s,\pi_h^\ast(s)) = (\mathcal{L}_h^{\pi_h^\ast}V_{h+1})(s) - \varepsilon_{h}(s,\pi_h^\ast(s)).
    \end{align*}
    If we apply \Cref{fact:quantum_minimum_finding} with $m = V_h^{\rm aux}(s)$ and unitary $W_s^{(h)}$, then after $K = O\big(\sqrt{A}\log\frac{1}{\delta_1}\big) = O\Big(\frac{\log(1/\delta_1)}{\sqrt{\mathrm{Pr}[X \geq V_h^{\rm aux}(s)]}}\Big)$ applications of $W_s^{(h)}$ and $W_s^{(h)\dagger}$, with probability at least $1-\delta_1$, the first outcome $V_h(s)$ of the second register ($I$ in $W_s^{(h)}$) and the first outcome $\pi_h(s)$ of the first register ($\mathsf{H}^{\log_2{A}}$ in $W_s^{(h)}$) satisfy
    \begin{align}\label{eq:approximation_finite}
        \forall s\in\mathcal{S}: \qquad 
        \begin{aligned}
            V_h(s) &\leq r_h(s,\pi_h(s)) + \mu_{h+1}(s,\pi_h(s)) = (\mathcal{L}_h^{\pi_h}V_{h+1})(s),\\
            V_h(s) &\geq V_h^{\rm aux}(s) \geq (\mathcal{L}_h^{\pi_h^\ast}V_{h+1})(s) - \varepsilon_{h}(s,\pi_h^\ast(s)).
        \end{aligned} 
    \end{align}
    %
    %
    Note that $V_{h} \leq \mathcal{L}_h^{\pi_h}V_{h+1} \leq \mathcal{L}_h^{\pi_h}V_{h+1}^{\pi,H} = V_h^{\pi,H}$ since $V_{h+1} \leq V_{h+1}^{\pi,H}$ by the induction hypothesis. 
    
    We are left with proving that $V^{\ast,H}_h(s) - \varepsilon \leq V_h(s)$. For such,
    \begin{align*}
        V_h^{\ast,H} - V_h &\leq \mathcal{L}_h^{\pi^\ast_h}V^{\ast,H}_{h+1} - \mathcal{L}_h^{\pi^\ast_h}V_{h+1} + \xi_{h} = \mathcal{P}_h^{\pi_h^\ast}(V^{\ast,H}_{h+1} - V_{h+1}) + \xi_{h},
    \end{align*}
    where we defined the function $\xi_h(s) = \varepsilon_{h}(s,\pi^\ast_h(s)) = 2\theta\sqrt{\sigma^\ast_{h+1}(s,\pi_h^\ast(s))} + 2\theta\varepsilon + 2\theta^{3/2} H$. Solving the above recursion with the boundary condition $V_{H-1} \equiv V_{H-1}^{\ast,H}$,
    \begin{align*}
        V_h^{\ast,H} - V_h \leq \sum_{h'=h}^{H-2}\left(\prod_{i=h}^{h'-1} \mathcal{P}_i^{\pi^\ast_i}\right) \xi_{h'}.
    \end{align*}
    Using that (where $\sigma_{h+1}^\ast = \mathcal{P}_h^{\pi_h^\ast} (V_{h+1}^{\ast,H})^2 - (\mathcal{P}_h^{\pi^\ast_h} V_{h+1}^{\ast,H})^2\in\mathscr{B}(\mathcal{S})$)
    \begin{align*}
        \left\|\sum_{h'=h}^{H-2} \left(\prod_{i=h}^{h'-1}\mathcal{P}_i^{\pi_i^\ast}\right) \sqrt{\sigma_{h'+1}^\ast}\right\|_\infty \leq H^{3/2}, \tag{by \Cref{lem:upper_bound_variance_infinite_horizon} with $\gamma = 1$ and $T=H$}
    \end{align*}
    and that $\|\sum_{h'=h}^{H-2} \prod_{i=h}^{h'-1} \mathcal{P}_i^{\pi^\ast_i}\|_\infty \leq H$, we get that
    \begin{align*}
        V_h^{\ast,H}(s) - V_h(s) \leq 2\theta H^{3/2} + 2\theta H\varepsilon + 2\theta^{3/2}H^2 
        = \varepsilon\left(\frac{2}{5} + \frac{2\varepsilon}{5\sqrt{H}} + \frac{2\sqrt{\varepsilon}}{5^{3/2}H^{1/4}} \right) < \varepsilon. \tag{$\theta = \frac{\varepsilon}{5H^{3/2}}$ and $\varepsilon\leq \sqrt{H}$}
    \end{align*}

    We now analyse the success probability of \Cref{algo:finite_horizon}.\footnote{See~\cite[Appendix~A]{chen2021quantum} for a similar argument.} Ideally, we would like to implement the unitary $\mathcal{U}_{s, {\rm ideal}}^{(h)}:|a\rangle|\bar{0}\rangle \mapsto |a\rangle|\Phi_{h}^{\rm good}(s,a)\rangle$. In practice, however, we implement the unitary 
    \begin{align*}
        \mathcal{U}_{s}^{(h)}:|a\rangle|\bar{0}\rangle \mapsto |a\rangle(\sqrt{1-\delta_2(s,a)}|\Phi_{h}^{\rm good}(s,a)\rangle + \sqrt{\delta_2(s,a)}|\Phi_{h}^{\rm bad}(s,a)\rangle)
    \end{align*}
    as mentioned above. Let $\delta_2 := \max_{s\in\mathcal{S}, a\in\mathcal{A}}\delta_2(s,a)$. Let a unitary $\mathcal{V}_s^{(h)}$ such that, for all $a\in\mathcal{A}$,
    \begin{enumerate}
        \item $\mathcal{V}_s^{(h)} |a\rangle(\sqrt{1-\delta_2(s,a)}|\Phi_{h}^{\rm good}(s,a)\rangle + \sqrt{\delta_2(s,a)}|\Phi_{h}^{\rm bad}(s,a)\rangle) = |a\rangle|\Phi_{h}^{\rm good}(s,a)(s,a)\rangle$;
        \item $\mathcal{V}_s^{(h)} |a\rangle(\sqrt{\delta_2(s,a)}|\Phi_{h}^{\rm good}(s,a)\rangle - \sqrt{1-\delta_2(s,a)}|\Phi_{h}^{\rm bad}(s,a)\rangle) = |a\rangle|\Phi_{h}^{\rm bad}(s,a)\rangle$ (notice that $\sqrt{\delta_2(s,a)}|\Phi_{h}^{\rm good}(s,a)\rangle - \sqrt{1-\delta_2(s,a)}|\Phi_{h}^{\rm bad}(s,a)\rangle$ is orthogonal to $\sqrt{1-\delta_2(s,a)}|\Phi_{h}^{\rm good}(s,a)\rangle + \sqrt{\delta_2(s,a)}|\Phi_{h}^{\rm bad}(s,a)\rangle$);
        \item For every $|\phi\rangle$ orthogonal to $\operatorname{span}\{|\Phi_{h}^{\rm good}(s,a)\rangle,|\Phi_{h}^{\rm bad}(s,a)\rangle\}$, $\mathcal{V}_s^{(h)}|a\rangle|\phi\rangle = |a\rangle|\phi\rangle$.
    \end{enumerate}
    Then 
    \begin{align*}
        \mathcal{V}_s^{(h)}\mathcal{U}_s^{(h)}|a\rangle|\bar{0}\rangle = \mathcal{V}_s^{(h)} |a\rangle(\sqrt{1-\delta_2(s,a)}|\Phi_{h}^{\rm good}(s,a)\rangle + \sqrt{\delta_2(s,a)}|\Phi_{h}^{\rm bad}(s,a)\rangle) = |a\rangle|\Phi_{h}^{\rm good}(s,a)\rangle.
    \end{align*}
    Therefore, we can just set $\mathcal{U}_{s, {\rm ideal}}^{(h)} = \mathcal{V}_s^{(h)}\mathcal{U}_s^{(h)}$ without loss of generality. As a consequence,
    \begin{align*}
        \|\mathcal{U}_{s}^{(h)} - \mathcal{U}_{s, {\rm ideal}}^{(h)}\| = \|I - \mathcal{V}_s^{(h)}\| \leq \sqrt{(1 - \sqrt{1-\delta_2})^2 + \delta_2} = \sqrt{2-2\sqrt{1-\delta_2}} \leq \sqrt{2\delta_2},
    \end{align*}
    using that $\sqrt{1-\delta_2} \geq 1 - \delta_2$. Moving on, the success probability of \Cref{fact:quantum_minimum_finding} is $1-\delta_1$ when employing $\mathcal{U}_{s, {\rm ideal}}^{(h)}$, for some $\delta_1\in(0,1)$. However, since it employs $\mathcal{U}_{s}^{(h)}$ instead, the success probability decreases by at most the spectral norm of the difference between the ``real'' and the ``ideal'' total unitaries. To be more precise, the ``ideal'' quantum maximum finding is a sequence of gates $W = U_1E_1U_2 E_2\cdots U_{K}E_K$, where $U_i \in \{\mathcal{U}_{s, {\rm ideal}}^{(h)},\mathcal{U}_{s, {\rm ideal}}^{(h)\dagger}\}$, $E_i$ is a circuit of elementary gates, and $K = c\sqrt{A}\log\frac{1}{\delta_1}$ with $c$ constant is the number of queries to $\mathcal{U}_{s, {\rm ideal}}^{(h)}$. The ``real'' implementation, on the other hand, is $\widetilde{W} = \widetilde{U}_1E_1\widetilde{U}_2 E_2\cdots \widetilde{U}_{K}E_K$, where $\widetilde{U}_i \in \{\mathcal{U}_{s}^{(h)},\mathcal{U}_{s}^{(h)\dagger}\}$. Then $\|W - \widetilde{W}\| \leq c\sqrt{A}\log\!\big(\frac{1}{\delta_1}\big)\|\mathcal{U}_{s, {\rm ideal}}^{(h)} - \mathcal{U}_{s}^{(h)}\| \leq c\sqrt{2\delta_2A}\log\frac{1}{\delta_1}$ and the failure probability is $\delta_1 + c\sqrt{2\delta_2A}\log\frac{1}{\delta_1}$. By taking $\delta_1 = O\big(\frac{\delta}{HS}\big)$ and $\delta_2 = O\big(\frac{\delta_1^2}{A\log^2(1/\delta_1)}\big)$, the failure probability in outputting $V_h(s)$ and $\pi_h(s)$ according to \eqref{eq:approximation_finite} is at most $\frac{\delta}{HS}$. By a usual union bound over all $s\in\mathcal{S}$ and $h\in[H]$, the failure probability is at most $\delta$.
    
    Regarding the query complexity, one call to the unitary $\mathcal{U}_s^{(h)}$ in order to compute $\widetilde{\mu}_{h+1}(s,a)$ and $\widetilde{\sigma}_{h+1}(s,a)$ in \eqref{eq:superposition_approximate_quantities} uses $O\big(\frac{1}{\theta}\log\frac{1}{\delta_2}\big) = \widetilde{O}\big(\frac{H^{3/2}}{\varepsilon}\log\frac{HSA}{\delta}\big)$ queries to $\mathcal{Q}_{p_h},\mathcal{Q}_{p_h}^\dagger$, while finding the maximum makes $O\big(\sqrt{A}\log\frac{1}{\delta_1}\big) = O\big(\sqrt{A}\log\frac{HS}{\delta}\big)$ queries to $\mathcal{U}_s^{(h)}$. Summing over all $\mathcal{S}\times[H]$, we obtain the stated query complexity.
\end{proof}

\section{Optimal Policies for Infinite-Horizon Discounted MDPs}
\label{app:infinite_horizon}

In this section, we describe our quantum algorithm (\Cref{algo:infinite_horizon}) for computing $\varepsilon$-optimal policies for infinite-horizon discounted MDPs.

\begin{theorem}\label{thr:quantum_infinite-horizon2}
    Let $\langle \mathcal{S},\mathcal{A}, \{p_t\}_{t\in\mathbb{N}}, \{r_t\}_{t\in\mathbb{N}},\gamma\rangle$ be a time-dependent infinite-horizon discounted MDP with $\gamma\in[0,1)$. Let $\delta\in(0,1)$, $\varepsilon\in(0,\sqrt{\Gamma}]$, and $T := \Gamma\lceil\ln(\Gamma/\varepsilon)\rceil$. {\rm \Cref{algo:infinite_horizon}} computes functions $\{V_t\}_{t\in[T]}\subset\mathscr{B}(\mathcal{S})$ and deterministic policy $\pi=(\pi_t)_{t\in[T]}$ such that, with probability $1-\delta$,
    \begin{align*}
        V^{\ast,\gamma}_t - (\varepsilon + \gamma^{T-t}\Gamma)\mathbf{1} \leq V^{\ast,T}_t - \varepsilon\mathbf{1} \leq V_t \leq V^{\pi,T}_t \leq  V^{\ast,T}_t \leq V_t^{\ast,\gamma} \quad\forall t\in[T].
    \end{align*}
    Moreover, let $\bar{\pi} = (\pi_0,\pi_1,\dots,\pi_{T-1},d,d,\dots)$ for any deterministic decision rule $d\in\mathcal{D}^{\rm D}$. Then
    \begin{align*}
        V^{\ast,\gamma}_t - (\varepsilon + \gamma^{T-t}\Gamma)\mathbf{1} \leq V_t \leq V_t^{\bar{\pi},\gamma} \leq V_t^{\ast,\gamma} \quad\forall t\in[T].
    \end{align*}
    In particular, $V^{\ast,\gamma}_0 - 2\varepsilon\mathbf{1} \leq V^{\bar{\pi},\gamma}_0 \leq  V^{\ast,\gamma}_0$. Its query complexity is (up to $\poly\log\log$ factors)
    \begin{align*}
        \widetilde{O}\bigg(\frac{\Gamma^{2.5}S\sqrt{A}}{\varepsilon}\log\left(\frac{\Gamma SA}{\delta} \right)\log\left(\frac{\Gamma S}{\delta} \right)\log\left(\frac{\Gamma}{\varepsilon}\right) \bigg).
    \end{align*}
\end{theorem}
\begin{algorithm}[t!]
    \caption{Quantum algorithm for optimal policies for infinite-horizon discounted MDPs}
    \label{algo:infinite_horizon}
    \begin{algorithmic}[1]  
    \Require Finite state space $\mathcal{S}$ and action space $\mathcal{A}$, discount factor $\gamma\in[0,1)$, quantum sampling access to $\{p_t\}_{t\in\mathbb{N}}$, quantum access to $\{r_t\}_{t\in\mathbb{N}}$, failure probability $\delta\in(0, 1)$, error $\varepsilon \in (0,\sqrt{\Gamma}]$.

    \Ensure $2\varepsilon$-optimal deterministic policy $\pi$. 

    \State $T \gets \Gamma\lceil\ln(\Gamma/\varepsilon)\rceil$ and $\theta \gets \frac{\varepsilon}{5\Gamma^{3/2}}$

    \State $V_{T-1}(s) \gets \max_{a\in\mathcal{A}}r_{T-1}(s,a)$ and $\pi_{T-1}(s) \gets \argmax_{a\in\mathcal{A}}r_{T-1}(s,a)$
        
        \For{$t=T-2,\dots,0$}

            \State Let $\mu_{t+1}(s,a) = \sum_{s'\in\mathcal{S}}p_t(s'|s,a) V_{t+1}(s')$ 
            
            \State Let $\sigma_{t+1}(s,a) = \sum_{s'\in\mathcal{S}}p_t(s'|s,a) V_{t+1}(s')^2 - \big(\sum_{s'\in\mathcal{S}}p_t(s'|s,a) V_{t+1}(s')\big)^2$ 

            \For{$s\in\mathcal{S}$}

            \State \parbox[t]{\dimexpr\linewidth-\algorithmicindent-\algorithmicindent}{Use \Cref{fact:quantum_mean_estimation_variance} to obtain a unitary $\mathcal{U}_{s}^{(t)}:|a\rangle|\bar{0}\rangle \mapsto |a\rangle|\Phi_{t}(s,a)\rangle$, where $|\Phi_{t}(s,a)\rangle$ is a quantum state that, upon measuring its first register in the computational basis, returns
            \begin{align*}
                Q_{t}(s,a) = r_t(s,a) + \gamma\widetilde{\mu}_{t+1}(s,a) - \gamma\theta\sqrt{\widetilde{\sigma}_{t+1}(s,a)} - \gamma\theta^{3/2} \Gamma
            \end{align*}
            for $\widetilde{\mu}_{t+1}(s, a)$ and $\widetilde{\sigma}_{t+1}(s, a)$ that satisfy, with failure probability $\widetilde{O}\big(\frac{\delta^2}{\Gamma^2 S^2 A}\big)$,
            \begin{align*}
                |\widetilde{\mu}_{t+1}(s,a) - \mu_{t+1}(s,a) | &\leq \theta\sqrt{\sigma_{t+1}(s,a)},\\
                |\widetilde{\sigma}_{t+1}(s,a) - \sigma_{t+1}(s,a) | &\leq \theta \Gamma^2
            \end{align*}}

            \State \parbox[t]{\dimexpr\linewidth-\algorithmicindent-\algorithmicindent}{Use \Cref{fact:quantum_minimum_finding} with $W_s^{(t)} = \mathcal{U}_s^{(t)}(\mathsf{H}^{\lceil\log_2{A}\rceil}\otimes I)$ (where $\mathsf{H}$ is the Hadamard gate) to obtain $V_t(s)$ and $\pi_t(s)$ such that, with probability $1-\frac{\delta}{TS}$,
			\begin{align*}
				V_t(s) &\leq r_t(s,\pi_t(s)) + \gamma\mu_{t+1}(s,\pi_t(s)),\\
				V_t(s) &\geq r_t(s,\pi_t(s)) + \gamma\mu_{t+1}(s,\pi_t(s)) - 2\gamma\theta\sqrt{\sigma_{t+1}(s,\pi_t(s))} - 2\gamma\theta^{3/2}\Gamma
			\end{align*}}

            \If{the MDP is time-independent}

            \If{$V_{t}(s) \leq V_{t+1}(s)$} $V_{t}(s) \gets V_{t+1}(s)$ and $\pi_{t}(s) \gets  \pi_{t+1}(s)$ \EndIf

            \EndIf
            \EndFor

        \EndFor
    
    \State {\bfseries return} $(\pi_0,\dots,\pi_{T-1})$
\end{algorithmic}
\end{algorithm}
\begin{proof}
    The proof is similar to \Cref{thr:quantum_finite-horizon2}. For any $T\in\mathbb{N}$, we prove by (backwards) induction that
    \begin{align*}
        V^{\ast,T}_t - \varepsilon\mathbf{1} \leq V_t \leq V^{\pi,T}_t \leq  V^{\ast,T}_t \leq V_t^{\ast,\gamma} \qquad\forall t\in[T].
    \end{align*}
    The inequalities $V^{\pi,T}_t \leq  V^{\ast,T}_t \leq V_t^{\ast,\gamma}$ and $V^{\pi,T}_t \leq V_t^{\bar{\pi},\gamma} \leq V_t^{\ast,\gamma}$ are straightforward, while the inequality $V^{\ast,\gamma}_t - (\varepsilon + \gamma^{T-t}\Gamma)\mathbf{1} \leq V_t $ follows from $V^{\ast,T}_t - \varepsilon\mathbf{1} \leq V_t$ and the fact that (let $\pi^\ast \in \arg\max_{\pi\in\Pi^{\rm D}} V^{\pi,\gamma}_t$ be an optimal deterministic policy)
    \begin{align*}
        V_t^{\ast,\gamma}(s) &= \mathbb{E}\left[\sum_{t'=t}^\infty \gamma^{t'-t} r_{t'}(s_{t'},\pi_{t'}^\ast(s_{t'})) ~\Bigg|~ s_t = s, s_i \sim p_i(\cdot|s_i,\pi_{i}^\ast(s_{i}))\right] \\
        &\leq \mathbb{E}\left[\sum_{t'=t}^{T-1} \gamma^{t'-t} r_{t'}(s_{t'},\pi_{t'}^\ast(s_{t'})) ~\Bigg|~ s_t = s, s_i \sim p_i(\cdot|s_i,\pi_{i}^\ast(s_{i}))\right] + \sum_{t'=T}^\infty \!\gamma^{t'-t} \tag{$r_{t'} \leq \mathbf{1}$} \\
        &\leq V_t^{\ast,T}(s) + \gamma^{T-t}\Gamma.
    \end{align*}
    Moreover, $\gamma^{T}\Gamma \leq e^{-T(1-\gamma)}\Gamma \leq \varepsilon$ if $T = \Gamma\lceil\ln(\Gamma/\varepsilon)\rceil$, using that $x\leq e^{-(1-x)}$ for all $x\in\mathbb{R}$. It thus only remains to prove $V^{\ast,T}_t - \varepsilon\mathbf{1} \leq V_t \leq V^{\pi,T}_t$.
    
    The base case $t=T-1$ of the induction is trivial since $V_{T-1}(s) = \max_{a\in\mathcal{A}}\{r_{T-1}(s,a)\} = V_{T-1}^{\ast,T}(s)$ and $\pi_{T-1}(s) = \argmax_{a\in\mathcal{A}}\{r_{T-1}(s,a)\}$. Assume then that \Cref{algo:infinite_horizon} has computed functions $V_{t+1},\dots,V_{T-1}$ and decision rules $\pi_{t+1},\dots,\pi_{T-1}$ (which form policy $\pi$) such that
    \begin{align*}
        V^{\ast,T}_{t'} - \varepsilon\mathbf{1} \leq V_{t'} \leq V^{\pi,T}_{t'} \leq  V^{\ast,T}_{t'} \leq V_{t'}^{\ast,\gamma} \qquad\forall t'=t+1,\dots,T-1,
    \end{align*}
    and consider the time step $t$. Let the true quantities
    \begin{align*}
        \mu_{t}(s,a) := (\mathcal{P}^a V_{t})(s) \qquad\text{and}\qquad 
        \sigma_{t}(s,a) := (\mathcal{P}^a V_{t}^2)(s) - (\mathcal{P}^a V_{t})^2(s).
    \end{align*}
    Let $\theta := \frac{\varepsilon}{5 \Gamma^{3/2}}$. Just like in \Cref{thr:quantum_finite-horizon2}, we can use \cref{fact:quantum_mean_estimation_variance} with oracles $\mathcal{Q}_{p_t},\mathcal{Q}_{p_t}^\dagger$ in order to construct a black-box oracle $\mathcal{U}_s^{(t)}:|a\rangle|\bar{0}\rangle \mapsto |a\rangle|\Phi_t(s,a)\rangle$, where $|\Phi_{t}(s,a)\rangle$ is a quantum state that, upon measuring its first register in the computational basis, returns
    \begin{align*}
        \forall(s,a)\in\mathcal{S}\times\mathcal{A}: \qquad  Q_{t}(s,a) = r_t(s,a) + \gamma\widetilde{\mu}_{t+1}(s,a) - \gamma\theta\sqrt{\widetilde{\sigma}_{t+1}(s,a)} - \gamma\theta^{3/2} \Gamma,
    \end{align*}
    where $\widetilde{\mu}_{t+1}(s, a)$ and $\widetilde{\sigma}_{t+1}(s, a)$ are quantities that satisfy, with probability $1-\delta_2(s,a)$. 
    \begin{align}\label{eq:approximations_infinite}
        \forall(s,a)\in\mathcal{S}\times\mathcal{A} : \qquad
        \begin{aligned}
            &|\widetilde{\mu}_{t+1}(s, a) - \mu_{t+1}(s,a) | \leq \theta\sqrt{\sigma_{t+1}(s,a)},\\
            &| \widetilde{\sigma}_{t+1}(s, a) - \sigma_{t+1}(s,a) | \leq \theta \Gamma^2,
        \end{aligned}
    \end{align}
    using that the standard deviation of $V^2_{t+1}$ is at most $\Gamma^2$ and thus one can compute approximations of $(\mathcal{P}^a_t V_{t+1}^2)(s)$ and $(\mathcal{P}^a_t V_{t+1})^2(s)$ with error at most $\theta \Gamma^2/2$ each. Similarly to \Cref{thr:quantum_finite-horizon2}, we can show that measuring $|\Phi_{t}(s,a)\rangle$ in the computational basis returns $Q_t(s,a)$ such that, with probability $1-\delta_2(s,a)$, 
    \begin{align}\label{eq:approximations_final_discounted}
        \forall(s,a)\in\mathcal{S}\times\mathcal{A} : \qquad
        \begin{aligned}
            &Q_{t}(s, a) \leq r_t(s,a) + \gamma\mu_{t+1}(s,a), \\
            &Q_{t}(s, a) \geq r_t(s,a) + \gamma\mu_{t+1}(s,a) - \gamma\varepsilon_{t}(s,a),
        \end{aligned}
    \end{align}
    where $\varepsilon_{t}(s,a) := 2\theta\sqrt{\sigma^\ast_{t+1}(s,a)} + 2\theta\varepsilon + 2\theta^{3/2} \Gamma$.

    The use of \Cref{fact:quantum_minimum_finding} in order to find the maximum over $a\in\mathcal{A}$ is done just like in \Cref{thr:quantum_finite-horizon2}. Write
    \begin{align*}
        |\Phi_{t}(s,a)\rangle = \sqrt{1-\delta_2(s,a)}|\Phi_{t}^{\rm good}(s,a)\rangle + \sqrt{\delta_2(s,a)}|\Phi_{t}^{\rm bad}(s,a)\rangle,
    \end{align*}
    where $|\Phi_{t}^{\rm good}(s,a)\rangle$ is the superposition over numbers that satisfy \eqref{eq:approximations_final_discounted}, while $|\Phi_{t}^{\rm bad}(s,a)\rangle$ is the superposition over numbers that do not satisfy \eqref{eq:approximations_final_discounted}. Assume for now access to a unitary $\mathcal{U}_{s,\text{ideal}}^{(t)}:|a\rangle|\bar{0}\rangle \mapsto |a\rangle|\Phi_{t}^{\rm good}(s,a)\rangle$ and let $W_s^{(t)} = \mathcal{U}_{s,\text{ideal}}^{(t)}(H^{\otimes \log_2{A}}\otimes I)$. Let
    \begin{align*}
        V_t^{\rm aux}(s) &= \max_{a\in\mathcal{A}}\{r_t(s,a) + \gamma\mu_{t+1}(s,a) - \gamma\varepsilon_{t}(s,a)\}.
    \end{align*}
    If we apply \Cref{fact:quantum_minimum_finding} with $m = V_t^{\rm aux}(s)$ and unitary $W_s^{(t)}$, then after $K = O\big(\sqrt{A}\log\frac{1}{\delta_1}\big) = O\Big(\frac{\log(1/\delta_1)}{\sqrt{\mathrm{Pr}[X \geq V_t^{\rm aux}(s)]}}\Big)$ applications of $W_s^{(t)}$ and $W_s^{(t)\dagger}$, with probability at least $1-\delta_1$, the first outcome $V_t(s)$ of the second register and the first outcome $\pi_t(s)$ of the first register satisfy
    \begin{align}\label{eq:auxiliar_equation_discounted}
        \forall s\in\mathcal{S}: \qquad 
        \begin{aligned}
            V_t(s) &\leq r_t(s,\pi_t(s)) + \gamma\mu_{t+1}(s,\pi_t(s)) = (\mathcal{L}_t^{\pi_t}V_{t+1})(s),\\
            V_t(s) &\geq V_t^{\rm aux}(s) \geq (\mathcal{L}_t^{\pi_t^\ast}V_{t+1})(s) - \gamma\varepsilon_{t}(s,\pi_t^\ast(s)).
        \end{aligned}
    \end{align}
    Note that $V_t \leq \mathcal{L}_t^{\pi_t}V_{t+1} \leq \mathcal{L}_t^{\pi_t}V_{t+1}^{\pi, T} =V_{t}^{\pi, T} $ since $V_{t+1} \leq V_{t+1}^{\pi, T}$ by the induction hypothesis.
    We move on to proving that $V^{\ast,T}_{t} - \varepsilon\mathbf{1} \leq V_{t}$. For such,
    \begin{align*}
        V_{t}^{\ast,T} - V_{t} &\leq \mathcal{L}_t^{\pi_{t}^\ast}V^{\ast,T}_{t+1} - \mathcal{L}_t^{\pi_{t}^\ast}V_{t+1} + \gamma\xi_{t} = \gamma \mathcal{P}_t^{\pi_{t}^\ast}(V^{\ast,T}_{t+1} - V_{t+1}) + \gamma\xi_{t}, \tag{by \eqref{eq:auxiliar_equation_discounted} and $V_{t}^{\ast,T} = \mathcal{L}_t^{\pi_{t}^\ast}V_{t+1}^{\ast,T}$}
    \end{align*}
    where we defined the function $\xi_{t}(s) = \varepsilon_t(s,\pi_t^\ast(s)) = 2\theta\sqrt{\sigma_{t+1}^\ast(s,\pi_{t}^\ast(s))} + 2\theta\varepsilon + 2\theta^{3/2} \Gamma$.
    Solving the above recursion with the boundary condition $V_{T-1}^{\ast,T} \equiv V_{T-1}$,
    \begin{align*}
        V_{t}^{\ast,T} - V_{t} &\leq \sum_{t'=t}^{T-2} \left(\prod_{i=t}^{t'-1} \gamma \mathcal{P}_i^{\pi_i^\ast}\right) \gamma\xi_{t'}. 
    \end{align*}
    Using that (where $\sigma^\ast_t = \mathcal{P}_t^{\pi_{t}^\ast} (V_{t+1}^{\ast,T})^2 - (\mathcal{P}_t^{\pi_{t}^\ast} V_{t+1}^{\ast,T})^2\in\mathscr{B}(\mathcal{S})$)
    \begin{align*}
        \left\|\sum_{t'=t}^{T-2} \left(\prod_{i=t}^{t'-1} \gamma \mathcal{P}_i^{\pi_i^\ast}\right) \gamma\sqrt{\sigma_{t'+1}^\ast}\right\|_\infty  \leq \Gamma^{3/2}, \tag{by \Cref{lem:upper_bound_variance_infinite_horizon}}
    \end{align*}
    and that $\|\sum_{t'=t}^{T-2} \prod_{i=t}^{t'-1} \gamma \mathcal{P}_i^{\pi_i^\ast}\|_\infty \leq \sum_{t'=t}^{T-2} \gamma^{t'-t} \leq \Gamma$, then finally
    \begin{align*}
        V_{t}^{\ast,\gamma}(s) - V_{t}(s) \leq 2\theta \Gamma^{3/2} + 2\theta \Gamma\varepsilon + 2\theta^{3/2}\Gamma^2 
        = \varepsilon\left(\frac{2}{5} + \frac{2\varepsilon}{5\sqrt{\Gamma}} + \frac{2\sqrt{\varepsilon}}{5^{3/2}\Gamma^{1/4}} \right) < \varepsilon. \tag{$\theta = \frac{\varepsilon}{5\Gamma^{3/2}}$ and $\varepsilon\leq \sqrt{\Gamma}$}
    \end{align*}

    The error analysis of \Cref{algo:infinite_horizon} is very similar to \Cref{algo:finite_horizon}. The failure probability of \Cref{fact:quantum_minimum_finding} with $\mathcal{U}_s^{(t)}$ instead of $\mathcal{U}_{s,\text{ideal}}^{(t)}$ is $\delta_1 + c\sqrt{2\delta_2 A}\log\frac{1}{\delta_1}$, where $\delta_1$ is the error associated with \Cref{fact:quantum_minimum_finding} and $\delta_2 = \max_{s\in\mathcal{S},a\in\mathcal{A}}\delta_2(s,a)$ is the maximum error behind quantum mean estimation (\Cref{fact:quantum_mean_estimation_variance}). Choosing $\delta_1 = O\big(\frac{\delta}{TS}\big)$ and $\delta_2 = O\big(\frac{\delta_1^2}{A\log^2(1/\delta_1)}\big)$ leads to outputting $V_t(s)$ and $\pi_t(s)$ from \eqref{eq:auxiliar_equation_discounted} with probability at least $1-\frac{\delta}{TS}$. A union bound over all $(s,t)\in\mathcal{S}\times[T]$ leads to a failure probability at most $\delta$.

    Regarding the query complexity, for each $s\in\mathcal{S}$, one call to the unitary $\mathcal{U}_s^{(t)}$ in order to compute $\widetilde{\mu}_{t}(s,a)$ and $\widetilde{\sigma}_{t}(s,a)$ in \eqref{eq:approximations_infinite} uses $O\big(\frac{1}{\theta}\log\frac{1}{\delta_2}\big) = \widetilde{O}\big(\frac{\Gamma^{3/2}}{\varepsilon}\log\frac{\Gamma SA}{\delta}\big)$ queries to $\mathcal{Q}_{p_t}$ and $\mathcal{Q}_{p_t}^\dagger$, while finding the maximum requires $O\big(\sqrt{A}\log\frac{1}{\delta_1}\big) = \widetilde{O}\big(\sqrt{A}\log\frac{\Gamma S}{\delta}\big)$ queries to $\mathcal{U}_s^{(t)}$. Summing over all $\mathcal{S}\times[T]$, we obtain the stated query complexity.
\end{proof}

\begin{corollary}
    Let $\langle \mathcal{S},\mathcal{A}, p, r,\gamma\rangle$ be a time-independent infinite-horizon discounted MDP with $\gamma\in[0,1)$. Let $\delta\in(0,1)$, $\varepsilon\in(0,1]$, and $T := \Gamma\lceil\ln(\Gamma/\varepsilon)\rceil$. {\rm \Cref{algo:infinite_horizon}} computes functions $\{V_t\}_{t\in[T]}\subset\mathscr{B}(\mathcal{S})$ and deterministic policy $\pi=(\pi_t)_{t\in[T]}$ such that, with probability $1-\delta$,
    \begin{align*}
        V^{\ast,\gamma}_t - (\varepsilon + \Gamma\gamma^{T-t})\mathbf{1} \leq V_t \leq V_t^{\pi_t^\infty,\gamma} \leq V_t^{\ast,\gamma} \qquad\forall t\in[T].
    \end{align*}
    In particular, $V^{\ast,\gamma}_0 - 2\varepsilon\mathbf{1} \leq V^{\pi_{0}^\infty,\gamma}_0 \leq  V^{\ast,\gamma}_\infty$. Its query complexity is (up to $\poly\log\log$ factors)
    \begin{align*}
        \widetilde{O}\bigg(\frac{\Gamma^{2.5}S\sqrt{A}}{\varepsilon}\log\left(\frac{\Gamma SA}{\delta} \right)\log\left(\frac{\Gamma S}{\delta} \right)\log\left(\frac{\Gamma}{\varepsilon}\right) \bigg).
    \end{align*}
\end{corollary}
\begin{proof}
    In the time-independent setting, the optimal policy is deterministic and stationary. In order for {\rm \Cref{algo:infinite_horizon}} to output a deterministic stationary policy that approximates $V_t^{\ast,\gamma}$, we need to include a monotonicity step. To be more precise, given the outputs $V_t'(s)$ and $\pi'_t(s)$ of {\rm \Cref{fact:quantum_minimum_finding}} that satisfy {\rm \eqref{eq:auxiliar_equation_discounted}}, i.e.,
    \begin{align}\label{eq:remark}
        \forall s\in\mathcal{S}: \qquad 
        \begin{aligned}
            V_t'(s) &\leq (\mathcal{L}^{\pi'_t}V_{t+1})(s),\\
            V_t'(s) &\geq (\mathcal{L}^{\pi_t^\ast}V_{t+1})(s) - \gamma\varepsilon_{t}(s,\pi_t^\ast(s)),
        \end{aligned}
    \end{align}
    we define the actual value functions and decision rule at step $t$ as
    \begin{align*}
        V_{t}(s) = \max\{V_{t}'(s),V_{t+1}(s)\} \quad\text{and}\quad \pi_{t}(s) = \begin{cases}
             \pi_{t}'(s) &\text{if}~ V_{t}'(s)\geq V_{t+1}(s),\\
             \pi_{t+1}(s) &\text{if}~ V_{t}'(s) < V_{t+1}(s).
        \end{cases}
    \end{align*}
    Using this step, we can then prove that $V_t \leq \mathcal{L}^{\pi_t}V_t$. There are two cases to analyse. 
    
    \textbf{Case I:} $V'_{t}(s) \geq V_{t+1}(s)$. Then
    \begin{align*}
        V_{t}(s) = V_t'(s) \leq (\mathcal{L}^{\pi_{t}} V_{t+1})(s) \leq (\mathcal{L}^{\pi_{t}} V_{t})(s), 
    \end{align*}
    using {\rm \eqref{eq:remark}} and $V_{t+1} \leq V_t$.
     
    \textbf{Case II:} $V'_{t}(s) < V_{t+1}(s)$. Then
    \begin{align*}
        V_{t}(s) = V_{t+1}(s) \leq (\mathcal{L}^{\pi_{t+1}} V_{t+1})(s) = (\mathcal{L}^{\pi_{t}} V_{t+1})(s) \leq (\mathcal{L}^{\pi_{t}} V_{t})(s), 
    \end{align*}
    using that $V_{t} \geq V_{t+1}$ and $\pi_{t+1}(s) = \pi_t(s)$ and the induction hypothesis in $V_{t+1} \leq \mathcal{L}^{\pi_{t+1}} V_{t+1}$. 
    
    This proves that $V_{t} \leq \mathcal{L}^{\pi_{t}} V_{t}$, which by the monotonicity of the Bellman operator $\mathcal{L}^{\pi_{t}}$ and the fact that $V_t^{\pi_{t}^\infty,\gamma}$ is its fixed point, leads to $V_{t} \leq (\mathcal{L}^{\pi_{t}})^\infty V_{t} = V_t^{\pi_{t}^\infty,\gamma} \leq V_t^{\ast,\gamma}$. In particular, $V_0^{\ast,\gamma} - 2\varepsilon\mathbf{1} \leq V_0 \leq V_0^{\pi_0^\infty,\gamma} \leq V_0^{\ast,\gamma}$ and the deterministic stationary policy $\pi_0^\infty$ is $2\varepsilon$-optimal.
\end{proof}

\bibliographystyle{plain}
\bibliography{bibliography}

\end{document}